\documentclass[reqno,a4paper,10pt]{amsart}
\usepackage[a4paper, centering]{geometry}
\usepackage[utf8]{inputenc} 
\usepackage[english]{babel}
\usepackage[T1]{fontenc}
\usepackage{lmodern}
\usepackage{textcomp} 
\usepackage{caption}
\usepackage{subcaption}
\usepackage{enumerate}
\usepackage{amssymb, amsmath}
\usepackage{mathrsfs,dsfont}
\usepackage{amscd}
\usepackage{bbm}
\usepackage{amsthm}
\usepackage{cite}
\usepackage{esint }
\usepackage[mathcal]{euscript}
\usepackage[active]{srcltx}
\usepackage{verbatim}
\usepackage[colorlinks,linkcolor={blue},citecolor={blue},urlcolor={black}]{hyperref}
\usepackage[]{changebar}
\usepackage{xcolor}
\usepackage{mathtools}
\usepackage{url}            
\usepackage{booktabs}       
\usepackage{amsfonts}       
\usepackage{nicefrac}       
\usepackage{microtype}      
\usepackage{lipsum}
\usepackage{fancyhdr}       
\usepackage{graphicx}       
\graphicspath{{media/}}     

\usepackage{bookmark} 
\usepackage[a4paper, centering]{geometry}
\usepackage{epstopdf}
\usepackage{amscd}
\usepackage{color}
\usepackage{graphicx}

\usepackage{yfonts}
\usepackage{fullpage}
\usepackage{comment}
\usepackage[braket,qm]{qcircuit} 
\usepackage{ stmaryrd } 

\usepackage[many]{tcolorbox}
\newtcolorbox{mybox}{enhanced,colback=red!5!white, colframe=red!75!black, width=\textwidth,box align=center,halign=center,valign=center, center}
\usepackage[new]{old-arrows} 
\usepackage{mathabx} 

\usepackage{tabularx}
\usepackage{multirow}
\usepackage{makecell}

\newtheorem{thm}{Theorem}[section]

\newtheorem*{thm*}{Theorem}
\newtheorem{cor}{Corollary}[section]

\newtheorem{lem}{Lemma}[section]

\newtheorem{prop}{Proposition}[section]

\newtheorem*{prop*}{Proposition}

\theoremstyle{definition}
\newtheorem{defn}{Definition}[section]

\theoremstyle{remark}
\newtheorem{rem}{Remark}[section]

\numberwithin{equation}{section}

\def\N{{\mathbb N}}

\def\R{{\mathbb R}}
\def\C{{\mathbb C}}
\def\W{{\mathcal W}}

\def\norma #1{\left\lVert #1 \right\rVert}
\def\nm #1{ \left\langle #1 \right\rangle}

\def\P{{\mathscr{P}}}

\def\E{{\mathbb E}}

\def\H{{\mathscr H}}
\newcommand{\cP}{\mathbb{P}}

\def\Y{{\mathcal Y}}
\def\M{{\mathcal M}}
\def\L{{\mathcal L}}
\def\T{{\mathbb T}}
\def\y{{ \boldsymbol y}}

\def\X{{\mathscr X}}

\def\O{{\mathcal O}}

\def\D{{\mathcal D}}

\def\N{{\mathbb N}}

\def\R{{\mathbb R}}
\def\CC{{\mathbb C}}
\def\L{{\mathcal L}}

\def\H{{\mathcal H}}
\def\O{{\mathcal O}}
\def\E{{\mathbb E}}

\def\X{{\mathcal X}}
\def\Y{{\mathcal Y}}
\def\D{{\mathcal D}}

\def\cP{{\mathscr{P}}}
\def\norma #1{\left\lVert #1 \right\rVert}

\def\P{{\mathbb{P}}}
\def\C{{\mathscr{C}}}
\def\de{{\rm d}}

\def\nm #1{ \left\langle #1 \right\rangle}

\def\ov #1{\overline{#1}}

\definecolor{viola}{rgb}{0.3,0,0.7}
\definecolor{ciclamino}{rgb}{0.5,0,0.5}
\definecolor{rosso}{rgb}{0.8,0,0}

\newcommand{\beq}{\begin{equation}}
\newcommand{\eeq}{\end{equation}}
\newcommand{\bal}{\begin{aligned}}
\newcommand{\eal}{\end{aligned}}
\newcommand{\ben}{\begin{enumerate}}
\newcommand{\beni} {\begin{enumerate}[(i)]}
\newcommand{\een}{\end{enumerate}}
\newcommand{\bit}{\begin{itemize}}
\newcommand{\eit}{\end{itemize}}
\newcommand{\beqw}{\begin{equation*}}
\newcommand{\eeqw}{\end{equation*}}
\newcommand{\bex}{\begin{example}}
\newcommand{\eex}{\end{example}}
\newcommand{\bre}{\begin{example}}
\newcommand{\ere}{\end{example}}
\newcommand{\bma}{\begin{bmatrix}}
\newcommand{\ema}{\end{bmatrix}}

\renewcommand{\hat}{\widehat}

\makeatletter
\@namedef{subjclassname@2020}{%
  \textup{2020} Mathematics Subject Classification}
\makeatother

\title[MFG]{On a Central Limit Theorem and Sanov's principle for quantum neural networks}
\author[A.~Melchor Hernandez]{Anderson Melchor Hernandez}
\address[A.~Melchor Hernandez]{Dipartimento di Matematica, Via Zamboni, 33, 40126, Bologna (Italy)}
\email{anderson.melchor@unibo.it}
\date{\today}

\begin{document}

\begin{abstract}
In this work, we study the fluctuations of a Mixture of Experts (MoE) generated by a quantum neural network trained via gradient flow on supervised learning problems. Our main results establish the Central Limit Theorem (CLT), and Sanov's principle for an MoE as the number of experts diverges. We demonstrate that the fluctuations of the empirical measure of its parameters close to its corresponding limit probability measure solve a linear transport equation. As a byproduct, we show that the MoE converges to a limit function which solves an evolution equation governed by the neural tangent kernel associated with the quantum neural network.
\end{abstract}

\maketitle

\tableofcontents

\section{Introduction}
The growing interest in Quantum Machine Learning (QML) has catalyzed extensive research into the training dynamics and generalization behavior of Quantum Neural Networks (QNNs) \cite{de2019primer,schuld2015,pastorello2023concise,girardi2025,larocca2024review}. Among emerging branches of artificial intelligence, QML stands out for its integration of classical machine learning techniques with quantum computing paradigms \cite{de2019primer,schuld2015,pastorello2023concise}. The core idea behind QML is to harness quantum algorithms and quantum-mechanical phenomena—such as superposition, entanglement, and quantum parallelism—to enhance the performance of deep learning models \cite{biamonte2017}. One of the most prominent QML frameworks is the quantum neural network, which serves as the quantum counterpart to classical deep neural networks. The output of a QNN is defined as the expectation value of a quantum observable measured on the quantum state produced by a parameterized quantum circuit. This circuit comprises tunable one-qubit and two-qubit gates, with parameters encoding both input data and the trainable components of the model \cite{abedi2023,cerezo2021}. These parameters are typically optimized using gradient-based techniques to minimize a cost function and improve the circuit’s ability to process and analyze data \cite{kiani2022,schuld2018,schuld2021effect}.

In this paper, we investigate the fluctuations of a parametric model which may be implemented by a quantum neural network. Let us consider a finite set $\mathcal{X}$ of possible inputs, and let $\Theta$ represent the vector of circuit parameters. As a parametric model, we consider a function $x \mapsto f(\Theta, x)$. Suppose we have a training set $\{(x^{(i)},\,y^{(i)}) : i=1,\dots,n\}$, where $x^{(i)}\in\mathcal{X}$, and $y^{(i)}\in\mathbb{R}$. The aim of supervised learning is to find parameters $\Theta$ such that $f(\Theta, x)$ closely matches the labels $y^{(i)}$. It is a customary approach to minimize the empirical quadratic loss
\begin{equation}
\mathcal{L}(\Theta) = \frac{1}{2}\sum_{i=1}^n \left(f(\Theta,x^{(i)}) - y^{(i)}\right)^2
\end{equation}
using gradient-based optimization.
Recent research has investigated quantum neural networks for their potential to combine the computational advantages of quantum mechanics with the expressive power of deep learning models \cite{lloyd2020quantum}. Despite these promising developments, fundamental challenges persist—particularly in identifying and efficiently optimizing the parameters of quantum circuits \cite{cinelli2021var}.

Reference \cite{girardi2025} investigates quantum neural networks trained on supervised learning tasks, where the objective function is defined as the expected value of the sum of single-qubit observables across all qubits, normalized so that its variance at initialization scales as $\O(1)$. The authors analyze the training dynamics of these networks and establish their trainability in the infinite-width limit, provided the circuit depth is allowed to scale with the number of qubits $m$ and barren plateaus are avoided.
More specifically, \cite{girardi2025} demonstrates that the probability distribution of the function produced by the trained network converges in law to a Gaussian process, with analytically computable mean and covariance \cite[Theorem 4.15]{girardi2025}. Building on this result, \cite[Theorem 5.1, 5.2]{melchor2025quantitative} provides quantitative bounds for the convergence in terms of the Wasserstein-1 distance.
Subsequently, \cite{hernandez2025efficientclassicalcomputationneural} introduces an efficient classical algorithm for estimating the Neural Tangent Kernel (NTK) associated with quantum neural networks composed of arbitrary unitary operators from the Clifford group, interleaved with parametric gates defined by time evolution under Hamiltonians from the Pauli group. The key insight behind the algorithm is that the average over the initialization parameter distribution in the NTK definition can be exactly substituted by an average over just four discrete values, so that the resulting parametric gates are themselves Clifford operations. This substitution enables efficient classical simulation of the quantum circuit. 
When combined with prior results establishing the equivalence between wide quantum neural networks and Gaussian processes \cite{girardi2025,melchor2025quantitative}, this approach allows for efficient computation of the expected output of wide, trained quantum neural networks. Consequently, it demonstrates that such architectures cannot achieve quantum advantage.
In this paper, we pursue an alternative perspective on quantum neural network training by analyzing the fluctuations, grounded in the concept of the {\emph mean field limit}, which has been extensively investigated in the context of classical deep neural networks \cite{sirignano2021meanfieldanalysisdeep,nguyen2019meanfieldlimitlearning,nguyen2023rigorous,mei2019meanfieldtheorytwolayersneural,lu2020meanfieldanalysisdeepresnet}.
Roughly speaking, in the mean field limit, the empirical distribution of a network’s neurons is approximated by a smooth probability measure. From this perspective, each neuron can be interpreted as a particle evolving under a suitable gradient flow in parameter space. As the network width tends to infinity, one can derive a limiting description of the training dynamics in terms of a partial differential equation. More concretely, consider a multilayer feedforward neural network. When the number of neurons per layer grows to infinity—under appropriate scaling of weights and biases—the empirical distribution of neurons in each layer converges to a smooth measure. The function computed by the network then becomes governed by this limiting measure, allowing one to track its evolution through mean-field gradient flow equations. This framework offers a rigorous mathematical foundation for understanding why highly overparameterized networks can often avoid poor local minima, achieve low training error, and even exhibit strong generalization performance \cite{mei2019meanfieldtheorytwolayersneural,Rotskoff_2022,sirignano2021meanfieldanalysisdeep,araújo2019meanfieldlimitcertaindeep}.
Here, we explore the fluctuations of the mean-field viewpoint for quantum neural networks. We consider a uniform mixture of $N$ identical experts: 
\begin{equation}
    F_{N}(\Theta,x)\coloneqq\frac{1}{N}\sum_{i=1}^{N}f(\theta^{i},x),
\end{equation}
where $x$ is the input, $\theta^i$ are the parameters of the $i$-th expert, and $f$ is the model function of a single expert generated by a quantum neural network. That is, each expert is a parametric quantum circuit with model function $f$ defined in \eqref{model1} below. When the parameters $\Theta^{N}\coloneqq(\theta^{1},\ldots \theta^{N})$ are trained by gradient flow,  the result is an updated collection of parameter $\Theta_{t}^{N}\coloneqq(\theta_{t}^{1},\ldots, \theta_{t}^{N})$ that define a mixture of experts. This, in turn, yields an updated model function $F_{N}(\Theta_{t},x)$.  Previous work has shown that, as the number of experts $N \rightarrow +\infty$, the mixture of experts satisfies a principle of propagation of chaos, with the empirical measure $\mu_{\Theta_{t}^{N}}$ converging to the unique solution $\mu_t$ of a nonlinear continuity equation \cite{hernandez2025meanfieldlimitgeneralmixtures}.
In this work, we establish the Central Limit Theorem (CLT) for an MoE as the number of experts $N$ diverges. We demonstrate that the fluctuations of the empirical measure of its parameters, close to its corresponding limit probability measure, solve a linear transport equation. Specifically, we define the fluctuation term as $\delta_{t}^{N}\coloneqq \sqrt{N}(\mu_{\Theta_{t}^{N}}-\mu_{t})$. We prove that this sequence of measures converges weakly in law to a limiting fluctuation process $\delta_t$. In our approach, we make use of a functional setup suitable to prove that the limiting process $\delta_{t}$ can be characterized, and that it solves a transport equation.

\subsection{Our results}
In this paper, we prove the Central Limit Theorem (CLT) for the sequence $\delta_{t}^{N}$ in the limit of infinitely many experts, where the parameters of the QNN evolve according to the gradient flow associated with the minimization of a quadratic cost function in a supervised learning problem. We show that, at any fixed time $t>0$, the limiting process $\delta_{t}$ solves a linear transport equation. Without delving into all the technical details, we establish the following informal result (see \autoref{prop:initialcase}, and \autoref{mainprop1} for a formal statement).

\begin{thm}\label{thm:informal}
Consider the MoE $F(\Theta,x)$ induced by the set of $N$ identical experts $\{f(\theta^{i},x):i=1,\ldots,N\}$ where $x$ represents a generic input, $\Theta\coloneqq(\theta^{1},\ldots \theta^{N})$ is the vector of parameters supported on the Torus $\mathbb{T}^{d}$ of dimension $d\in\N$ with period $2\pi$, and $f$ is generated by a quantum neural network. Let each component of $\Theta$ be initialized by sampling it from the uniform distribution, and let then $\Theta$ be trained via gradient flow:
\begin{align}
\begin{aligned}
&\frac{\de \L(\Theta_{t}^{N})}{\de t}=-N\nabla_{\Theta}\L(\Theta_{t}^{N}),\\
&\L(\Theta_{t}^{N})\coloneqq \frac{1}{2}\sum_{i=1}^n \left(F(\Theta_{t}^{N},x^{(i)}) - y^{(i)}\right)^2.
\end{aligned}
\end{align}
Then, the stochastic process $\delta_{t}^{N}$ converges in law to a limit process $\delta_{t}$ as $N\rightarrow +\infty$, and whose limit solves the transport equation

\begin{align}\label{transporteq}
\begin{aligned}
&\frac{\de \delta_{t}(\theta)}{\de t}=-\nabla_{\theta}\cdot(\nabla\mathcal{V}(\theta,\mu_{t})\delta_{t}+\mu_{t}\nabla\mathcal{G}(\theta,\delta_{t})),\\
&\mathcal{G}(\theta,\mu)\coloneqq -\E_{\theta'\sim \mu}[W(\theta,\theta')],\, W(\theta,\theta')\coloneqq \sum_{j=1}^{n}f(\theta,x_{j})f(\theta',x_{j}),\\
&V(\theta^{\ell})\coloneqq  \sum_{j=1}^{n}f(\theta^{\ell},x_{j})y_{j},\,\mathcal{V}(\theta,\mu)\coloneqq V(\theta)- \E_{\theta'\sim \mu}[W(\theta,\theta')],
\end{aligned}
\end{align}
where $\mu_{t}$ is the limit of the empirical measure $\mu_{\Theta_{t}}^{N}$ associated to the MoE $F_{N}(\Theta_{t},x)$, and it solves the continuity equation

\begin{align}
&\frac{\de \mu_{t}(\theta)}{\de t}=-\nabla_{\theta}\cdot\left(\nabla\mathcal{V}(\theta,\mu_{t})\mu_{t}\right).
\end{align}
\end{thm}
Our CLT provides a quantitative statement that the fluctuations of $\mu_{\Theta_{t}^{N}}$ around $\mu_{t}$ are of order $\frac{1}{\sqrt{N}}$. This result contrasts with previous analyzes where the proximity between $\mu_{\Theta_{t}^{N}}$ and $\mu_{t}$ was measured using the Wasserstein distance of order 2, denoted as $\W_{2}$ (for details on the definition of this distance, see for instance \cite{ambrosio2008gradient,panaretos2020invitation,santambrogio2015optimal}). More precisely, the authors showed that there exists a positive constant $C$ depending on the dimension $d$ such that
\begin{align}\label{bound_2}
\E\W_{2}^{2}\left(\mu_{\Theta_{t}^{N}},\mu_{t}\right)\leq C2N^{-\frac{2}{d}}.
\end{align}
Our Central Limit Theorem provides a tighter approximation rate, $O(N^{-\frac{1}{2}})$, for the fluctuations of the empirical measure. Crucially, in the high-dimensional setting where the number of parameters $d$ is greater than $4$, the term $N^{-2/d}$ dominates the convergence rate of $\W_{2}^{2}$, making it slower than the $N^{-1/2}$ rate established by our CLT since $2/d < 1/2$. Therefore, our result establishes a superior asymptotic rate for the convergence of the empirical measure, especially in the parameter-rich regime typical of practical quantum neural networks. Moreover, we characterize the limit of the fluctuations around $\mu_{t}$ as a process solving the equation \eqref{transporteq}. Furthermore, as a byproduct of the CLT for the empirical measure, we also establish the convergence of the MoE model function. We define the fluctuations of the MoE output $F_{N}(\Theta_{t},x)$ around its mean-field limit $f_{t}(x)\coloneqq \E_{\mu_{t}}[f(\theta,x)]$ as
\begin{align}
\psi_{N}(t,x)\coloneqq \sqrt{N}(F_{N}(\Theta_{t},x)-f_{t}(x))
\end{align}
and show that, in the limit $N\rightarrow +\infty$, this fluctuation process converges in law:
\begin{align}
\psi_{N}(t,x) \rightarrow \int_{\T^{d}}f(\theta,x)\delta_{t}(\theta) \quad \text{in law.}
\end{align}
This result confirms that the $O(1/\sqrt{N})$ rate of convergence propagates from the parameter measure to the model's output, providing the finite-size stochastic effects of the quantum neural network mixture. On the other hand, we establish a large deviation principle for the law of the process $\mu_{\Theta_{\cdot}^{N}}:[0,T]\ni t\mapsto \mu_{\Theta_{t}^{N}}$. We derive this result by introducing a suitable metric on the space of curves defined on the interval $[0,T]$ with values in the space of probability measures on $\mathbb T^{d}$. Let us now informally state our large deviations result (see \autoref{thm:LDP} for a precise formulation).

\begin{thm}\label{thm:informal_LDP}
Assume the same setting as in \autoref{thm:informal}. Then the law of the empirical process
$\mu_{\Theta_{\cdot}^{N}}$ satisfies a large deviation principle with speed $N$ and an explicit rate functional given by the relative entropy. Furthermore, there exists a
sequence $(\overline\Theta_{t}^{N})_{t\in[0,T]}$ of independent random variables such that
$\mu_{\Theta_{\cdot}^{N}}$ and $\mu_{\overline\Theta_{\cdot}^{N}}$ are exponentially
equivalent at speed $(\log N)^{\beta}$ for $0<\beta<1$.
\end{thm}
Let us briefly comment on this result. We emphasize that it depends on the weak convergence
of $\mu_{\Theta_{\cdot}^{N}}$. Thus, we first identify the appropriate space in which this
weak convergence holds, in order to apply effectively the abstract result proved in
\cite{Gartner1987}. On the other hand, the exponential equivalence between $\mu_{\Theta_{\cdot}^{N}}$ and $\mu_{\overline\Theta_{\cdot}^{N}}$ at speed $(\log N)^{\beta}$ follows from the explicit bound \eqref{bound_2}, which prevents us from obtaining the natural speed of order $N$ usually expected in this context (see, for instance, \cite[Definition 4.2.10]{dembo2011}).

This work is organized as follows. In \autoref{sec:prelim}, we set the notation of the paper, and we state our functional setup. In \autoref{sec:mainresults}, we state our main results \autoref{prop:initialcase}, \autoref{mainprop1}, and \autoref{thm:LDP}. In \autoref{proofs}, we provide the proofs of our main results. Finally, in \autoref{sec:concl}, we present some concluding remarks and discuss potential directions for future research.

\section{Preliminaries and notation}\label{sec:prelim}
In what follows, let us first recall the notion of large deviations used later on.
\subsection{Large deviations}
In what follows, let $X$ be a Haussdorf topological space; $(\mu^{N})$ a sequence of probability measures on $X$, and $(\gamma^{N})$ a sequence of numbers tending to $+\infty$ as $N\rightarrow +\infty$.

\begin{defn}
Let $L: X\rightarrow [0,+\infty]$ be a functional. We say that the triple $(X,(\mu^{N}), (\gamma^{N}))$ satisfies the principle of large deviations with the functional rate $L$, if the following conditions are satisfied:
\begin{itemize}
    \item[1.] for each open set $G\subseteq X$ 
    \begin{align}\label{lower_b}
        \liminf_{N\rightarrow}\frac{1}{\gamma^{N}}\log\mu^{N}(G)\geq -\inf_{x\in G}L(x);
    \end{align}
    \item[2.] for each closed subset $C\subseteq X$
    \begin{align}\label{upper_b}
      \limsup_{N\rightarrow}\frac{1}{\gamma^{N}}\log\mu^{N}(C)\leq -\inf_{x\in C}L(x); 
    \end{align}
    \item[3.] The sets $\Phi(s)\coloneqq \left\{x\in X: L(x)\leq s\right\}$, $s\geq 0$, are compact.
\end{itemize}
We recall that by convention, the infimum of the empty set is $+\infty$.
\end{defn}
Let us now recall the following abstract result due to Dawson and G\"{a}rtnert \cite{Gartner1987}. Let $X$ be a real vector space, and $X^{\ast}$ its algebraic dual space. Let us consider $Y\subseteq X^{\ast}$ equipped with the subspace topology induced by the weak$^\ast$ topology $\sigma(X^{\ast},X)$. Let $(\mu^{N})$ be aequence of probability measures on $Y$, and let $\gamma^{N}$ a sequence of positive numbers such that $\lim_{N\rightarrow +\infty}\gamma^{N}=+\infty$. Lastly, let us denote $\nm{x^{\ast},x}$ the duality between a linear functional $x^{\ast}\in X^{\ast}$, and $x\in X$.
\begin{thm}[{\cite[Theorem 3.4]{Gartner1987}}]\label{thm:LDP}
Suppose that the following conditions are satisfied:
\begin{itemize}
    \item[a.] for each $x\in X$ the limit 
\begin{align}\label{lim_H}
  \H(x)=\lim_{N\rightarrow+\infty}\frac{1}{\gamma^{N}}\log\displaystyle\int_{Y}\exp\left(\gamma^{N}\nm{x^{\ast},x}\right)\mu^{N}(\de x^{\ast})  
\end{align}
exists and is finite;
\item[b.] the function $\H$ is G$\hat{a}$teaux differentiable, that is, the function $t\mapsto \H(x+ty)$ is differentiable for all $x,y\in X$. Let us set
\begin{align}\label{Fenchel}
L(x^{\ast})\coloneqq \sup_{x\in X}\left[\nm{x^{\ast},x}-\H(x)\right], \hskip 0,1cm x^{\ast}\in X^{\ast};
\end{align}
\item[3.]suppose that $\{x^{\ast}\in X^{\ast}: L(x^{\ast})<+\infty\}\subseteq Y$.
\end{itemize}
Then $(Y,(\mu^{N}),(\gamma^{N}))$ satisfies a large deviations principle. The corresponding rate functional is the restriction of $L$ to $Y$.
\end{thm}
\subsection{Sanov's theorem}
In what follows, we recall a version of Sanov's theorem. Let $X,\Y_{1},\ldots, \Y_{r}$ be Polish spaces. Let $\cP(X)$ be the space of probability measures on $X$ endowed with the topology of the weak convergence. Given a natural number $N\in\N$, let us define

\begin{align}\label{empiricalset}
 \cP^{N}(X)\coloneqq\left\{\frac{1}{N}\sum_{i=1}^{N}\delta_{x_{k}}: x_{1},\ldots x_{N}\in X\right\}    
\end{align}
where $\delta_{x}$ is the Dirac measure on $X$  with unit mass at $x$. We denote by $C_{b}(X)$ the space of bounded and continuous functions on $X$ equipped with the norm $\norma{\cdot}_{\infty}$. Let us consider $\left\{\P_{x}:x\in X\right\}$ be a Feller sequence on $\Y\coloneqq \Y_{1}\otimes\cdots \Y_{r}$, that is, a family of probability measures such that $x\mapsto \nm{\P_{x},F}$ is continuous with respect to $x$ for all $F\in C_{b}(\Y)$. Given an element $\cP^{N}(X)\ni \nu=N^{-1}\sum_{i=1}^{N}\delta_{x_{k}}$, let us consider the probability measure $\hat{\P}\coloneqq\P_{x_{1}}\otimes\cdots\otimes\P_{x_{n}}$ under which we consider the probabily measure $Q_{\nu}^{N}$ which is the image of $\hat{\P}$ with respect to the map 

\begin{align}\label{map}
\Y^{N}\ni (y_{1},\ldots y_{N})\mapsto \left(\frac{1}{N}\sum_{k=1}^{N}\delta_{y_{k}^{(1)}},\ldots,\frac{1}{N}\sum_{k=1}^{N}\delta_{y_{k}^{(r)}}\right)\in \cP(\Y_{1})\otimes\cdots\otimes \cP(\Y_{r}),   
\end{align}
 where $y_{k}=(y_{k}^{(1)},\ldots, y_{k}^{(r)})$. Let us now state the following version of Sanov's theorem proved in \cite{Gartner1987}.

\begin{thm}[{\cite[Theorem 3.5]{Gartner1987}}]\label{thm:gen_sanov}
Let us give $\nu^{N}\in \cP^{N}(X)$ weakly converging to some $\nu\in \cP(X)$. Then $\left(\cP(\Y_{1})\otimes\cdots\otimes\cP(\Y_{r}), (Q_{\nu^{N}}^{N}),N\right)$ satisfies a large deviations principle with rate functional

\begin{align}
\begin{aligned}
L_{\nu}(\mu_{1},\ldots,\mu_{r})= \sup_{(g_{1},\ldots,g_{r})\in C_{b}(\Y_{1})\times\cdots\times C_{b}(\Y_{r})}&\left[\sum_{i=1}^{r}\int_{\Y_{i}}g_{i}(z)\mu_{i}(\de z)\right.\\
&\hskip -0,9cm\left.-\int_{X}\nu(\de x)\log\int_{\Y}\P_{x}\left(\de y^{(1)},\ldots, \de y^{(r)}\right)\exp\left(\sum_{i=1}^{r}g_{i}(y^{(i)})\right)\right],
\end{aligned}
\end{align}
where $(\mu_{1},\ldots,\mu_{r})\in \cP(\Y_{1})\times\cdots\times\cP(\Y_{r})$.
\end{thm}
Based on \autoref{thm:gen_sanov}, we will then proved that the empirical sequence $\mu_{\Theta_{t}^{N}}$ admits a large deviations principle, and the rate functional can be explicitly computed. Later on, we will need of the following setup. Let $T>0$, and $X$ be a generic metric space. Let us denote by $C([0,T],X)$ the space of continuous curves $\y:[0,T]\rightarrow X$ endowed with the metric 

\begin{align}\label{metric_X}
\de_{T,X}(\y_{1},\y_{2})\coloneqq\sup_{t\in [0,T]}\de_{X}(\y_{1}(t),\y_{2}(t)), \hskip 0,2cm \text{for all $\y_{1},\y_{2}\in C([0,T],X)$.}
\end{align}
In what follows, we introduce our notation about quantum neural networks.
\subsection{Training data}\label{sub:trdata}
We denote by $\X$ be the feature space, \emph{i.e.}, the set of all the possible inputs, and we let $\mathbb{R}$ be the output space. As a training set, we consider
Let
\begin{equation}
    \D \coloneqq\left\{(x^{(i)},y^{(i)}):i=1,\ldots,n\right\}\subset \X\times\mathbb{R},
\end{equation}
and we set $n=\vert \D\vert$ to be the cardinality of $\D$. This notation agrees with the one used in the previous works \cite{girardi2025,melchor2025quantitative,hernandez2025meanfieldlimitgeneralmixtures}.
\subsection{Quantum neural networks}
In what follows, we denote by $m\in \N$ the number of qubits of the quantum circuit implementing each expert. 
Let $\CC^{2}$ be the Hilbert space of a single qubit. We consider an observable $\O$ on the Hilbert space $\H=\left(\CC^{2}\right)^{\otimes m}$ with $\|\O\|\le1$.
The model function of each expert is then
    \begin{align}\label{model1}
    f(\theta,x)&\coloneqq \bra{0^{m}}U^{\dag}(\theta,x)\O\,U(\theta,x)\ket{0^{m}},    
    \end{align}
where $U$ is a unitary operation given by the expression

\begin{align}\label{uformula}
U(\theta,x)\coloneqq V_{d}(x)e^{-\frac{i\theta_{d}\mathcal{G}_{d}}{2}}V_{d-1}(x)\cdots V_{1}e^{-\frac{i\theta_{1}\mathcal{G}_{1}}{2}}V_{0}(x)    
\end{align}
where $V_{j}(x)\in \mathcal{L}(\H)$, $j=0,\ldots,d$ are unitary operations, and  $\mathcal{G}_{i}$ hermitian operators such that $\norma{\mathcal{G}_{i}}\leq 1$ for all $i=1,\ldots,d$. Now, we consider the quantum neural network defined as a classical mixture of $N$ identical quantum experts with independent parameters.
The corresponding model function is given by
\begin{equation}\label{eq:quantumMoE}
    F_{N}(\Theta,x)=\frac{1}{N}\sum_{i=1}^N \bra{0^{m}}U^{\dag}(\theta^i,x)\O\,U(\theta^i,x)\ket{0^{m}}.  
\end{equation}  
Furthermore, we consider the mean square error that is defined as

\begin{align}\label{costfunct1}
 \L(\Theta)\coloneqq \frac{1}{2}\sum_{j
=1}^{n}\left(F_{N}(\Theta,x_{j})-y_{j}\right)^{2}.  
\end{align}   
Expanding the squared term of the cost function, we can express $\L(\Theta)$ as follows:

\begin{align}
 \L(\Theta)&=\frac{1}{2}\sum_{j=1}^{n}\left(F_{N}(\Theta,x_{j})-y_{j}\right)^{2}\\
&=\frac{1}{2}\sum_{j=1}^{n}(y_{j})^{2} - \sum_{j=1}^{n}F_{N}(\Theta,x_{j})y_{j}+ \frac{1}{2}\sum_{j=1}^{n}\left\{\sum_{\ell,k=1}^{N}\frac{f(\theta^{\ell},x_{j})f(\theta^{k},x_{j})}{N^{2}}\right\}\\
&=\frac{1}{2}\sum_{j=1}^{n}(y_{j})^{2}-\frac{1}{N}\sum_{\ell=1}^{N}\sum_{j=1}^{n}f(\theta^{\ell},x_{j})y_{j}
+ \frac{1}{2N^{2}}\sum_{\ell,k=1}^{N}\sum_{j=1}^{n}f(\theta^{\ell},x_{j})f(\theta^{k},x_{j}).
\end{align}    
This expansion separates the total cost into three components: a constant term, a term linear in the expert functions $f(\theta,x)$, and and a term quadratic in them. This decomposition is instrumental for analyzing the mean-field limit of the system. Let us define the following terms to simplify the expression for $\L(\Theta)$ as its gradient:

\begin{align}
&V(\theta^{\ell})\coloneqq  \sum_{j=1}^{n}f(\theta^{\ell},x_{j})y_{j},\\
&W(\theta^{\ell},\theta^{k})\coloneqq \sum_{j=1}^{n}f(\theta^{\ell},x_{j})f(\theta^{k},x_{j}).
\end{align}
Here $V(\theta)$ represents the linear contribution in the expert functions weighted by the target labels, $y_{j}$. The term $W(\theta,\theta')$ represents the correlation (or interaction) that arises due to the different parameters $\theta$, $\theta'$ in the quadratic part in the experts. Then, the cost function becomes:

\begin{align}
\L(\Theta)=  \frac{1}{2}\sum_{j=1}^{n}(y_{j})^{2}-\frac{1}{N}\sum_{\ell=1}^{N}V(\theta^{\ell})+ \frac{1}{2N^{2}}\sum_{\ell,k=1}^{N}W(\theta^{\ell},\theta^{k}).
\end{align}
\subsection{Gradient flow and propagation of chaos}
The evolution of the composite parameter vector $\Theta$ follows the gradient flow equation 
\begin{align}\label{gradform1}
\frac{\de \Theta_{t}}{\de\,t}=-N\nabla_{\Theta}\L(\Theta_{t}), \hskip 0,2cm \Theta_{t}\in (\mathbb{T}^{d})^{N}.
\end{align}
This is the microscopic dynamic equation for the system of $N$ experts. The partial derivative (which is key to defining the gradient flow in \eqref{gradform1}) with respect to the parameters of the $i$-th expert, $\theta^{i}$, is

\begin{align}
\partial_{\theta^{i}}\L(\Theta)=-\frac{1}{N}\partial_{\theta^{i}}V(\theta^{i})+\frac{1}{N^{2}}\sum_{k=1}^{N}\partial_{\theta^{i}} W(\theta^{i},\theta^{k}).
\end{align}
which allows us to write the flow for the $i$-th expert's parameter, $\theta_{t}^{i}$, as
\begin{align}\label{gradeq3}
\frac{\de \theta_{t}^{i}}{\de t}=\partial_{\theta^{i}}V(\theta^{i})-\frac{1}{N}\sum_{k=1}^{N}\partial_{\theta^{i}}W(\theta^{i},\theta^{k}), \hskip 0,2cm i=1,\ldots,N.
\end{align}
To analyze the collective behavior, we introduce the empirical probability measure $\mu_{\Theta_{t}^{N}}$

\begin{align}\label{eq:empiricalmeasure}
\mu_{\Theta_{t}^{N}}\coloneqq \frac{1}{N}\sum_{i=1}^{N}\delta_{\theta_{t}^{i}}
\end{align}
where $\delta_{\theta_{t}^{i}}$ is the Dirac measure centered at $\theta_{t}^{i}$. This measure describes the distribution of the $N$ expert parameters at time $t$. By defining the mean-field vector field

\begin{align}\label{vectorfield}
\mathcal{V}(\theta,\mu)\coloneqq V(\theta)- \E_{\theta'\sim \mu}[W(\theta,\theta')],     
\end{align}

the dynamic equation for the $i$-th parameter $\theta_{t}^{i}$ can be written in a compact form as:

\begin{align}\label{edo}
\frac{\de \theta_{t}^{i}}{\de t}=\partial_{\theta}\mathcal{V}(\theta^{i},\mu_{\Theta_{t}^{N}}).
\end{align}
In this formulation, the term $\E_{\theta'\sim \mu}[W(\theta,\theta')]$ encapsulates the average interaction of the $i$-th expert with all other experts via the quadratic interaction term $W$. In the previous work \cite{hernandez2025meanfieldlimitgeneralmixtures}, the authors studied the evolution of the parameters of a quantum neural network through this flow equation. They showed that, in the limit $N\rightarrow+\infty$ (where $N$ is the number of experts) the mixture of experts satisfies a principle of propagation of chaos. Propagation of chaos is a crucial phenomenon in many-particle systems with mean-field interaction, where, as the number of particles grows, the behavior of any finite subset of particles becomes increasingly independent and identically distributed. This principle connects microscopic dynamics to macroscopic statistics, demonstrating that the collective evolution of a large system can often be effectively approximated by a limiting equation that governs the distribution of a single particle \cite{sznitman1991topics,graham1992mckean,graham1996asymptotic}.
More precisely, in the limit $N\rightarrow +\infty$, in \cite{hernandez2025meanfieldlimitgeneralmixtures} the authors proved that $\mu_{\Theta_{t}^{N}}$ converges to the unique solution $\mu_t$ of the nonlinear continuity equation

\begin{align}\label{FKP}
\frac{\de \mu_{t}(\theta)}{\de t}=-\nabla_{\theta}\cdot\left(\nabla\mathcal{V}(\theta,\mu_{t})\mu_{t}\right),
\end{align}
and $\mu_{t}$ is the probability law of the so-called McKean process $(\ov{\theta}_{t})_{t\geq 0}$ which solves the following nonlinear differential equation 

\begin{align}\label{MKV}
\de \ov{\theta}_{t}=\nabla\mathcal{V}(\ov{\theta}_{t},\mu_{t})\de t,
\end{align}
where $\mu_{t}={\mathrm Law}(\ov{\theta}_{t})$.
In what follows, we recall the notion of propagation of chaos as stated in \cite{hernandez2025meanfieldlimitgeneralmixtures}.

\begin{defn}\label{propchaos}
Let $T\in (0,+\infty]$, and $1\leq p<+\infty$. Propagation of chaos holds when for all $N\in\N$ there exist

\begin{itemize}
\item[$\bullet$] a system of particles $(\Theta_{t}^{N})_{t}=(\theta_{t}^{1},\ldots, \theta_{t}^{N})_{t}$ with law $\mu_{t}^{N}\in \M((\T^{d})^{N})$ at time $t\leq T$;
\item[$\bullet$] a system of independent stochastic processes $(\ov{\Theta}_{t}^{N})_{t}=(\ov{\theta}_{t}^{1},\ldots, \ov{\theta}_{t}^{N})$ with law $\mu_{t}^{\otimes N}\in \M((\T^{d})^{N})$ at time $t\leq T$, such that $\theta_{0}^{i}=\ov{\theta}_{0}^{i}$ $\P$-a.s. for $i=1,\ldots,N$;
\item[$\bullet$] a number $\varepsilon(N,T)$ such that $\varepsilon(N,T)\rightarrow 0$ as $N\rightarrow +\infty$,
\end{itemize}
such that (pathwise case)
\begin{align}
\frac{1}{N}\sum_{i=1}^{N}\E\left[\sup_{t\leq T}\norma{\theta_{t}^{i}-\ov{\theta}_{t}^{i}}_{1}^{p}\right]\leq \varepsilon(N,T),
\end{align}
or (pointwise case)
\begin{align}
\frac{1}{N}\sum_{i=1}^{N}\sup_{t\leq T}\E\left[\norma{\theta_{t}^{i}-\ov{\theta}_{t}^{i}}_{1}^{p}\right]\leq \varepsilon(N,T).
\end{align}
\end{defn}
\subsection{The Mean-Field Limit and the convergence of the mixture of experts}
The mixture of experts, $F_{N}(\Theta,x)$, is the expected value of a single expert's output $f(\theta,x)$ with respect to the empirical measure $\mu_{\Theta_{t}^{N}}$:

\begin{align}\label{meanvalue}
F_{N}(\Theta_{t},x)= \frac{1}{N}\sum_{i=1}^{N}f(\theta_{t}^{i},x)=\E_{\ov{\alpha}\sim \mu_{\Theta_{t}^{N}}}\left[f(\ov{\alpha},x)\right].
\end{align}
Consequently, in the limit $N\rightarrow+\infty$, the mixture of experts $F_{N}(\Theta,x)$ converges weakly to the expected value under the limiting measure $\mu_{t}$:

\begin{align}
F_{N}(\Theta_{t},x)\rightarrow \E_{\theta\sim \mu_{t}}[f(\theta,x)]\, \hskip 0,2cm \text{weakly as $N\rightarrow +\infty$.}    
\end{align}
In what follows, we denote this limit by $f_{t}(x)\coloneqq \E_{\theta\sim \mu_{t}}[f(\theta,x)$. By using the evolution equation \eqref{FKP} for $\mu_{t}$, we can derive the governing equation for the limiting function $f_{t}$:

\begin{align}
\frac{\de f_{t}(x)}{\de t}&=\frac{\de }{\de t}\int f(\theta,x)\mu_{t}(\de \theta)\\
&=\int f(\theta,x)\frac{\de \mu_{t}(\theta)}{\de t}\\
&= -\int f(\theta,x)\nabla_{\theta}\cdot\left(\nabla\mathcal{V}(\theta,\mu_{t})\mu_{t}\right)\\
&=\int \nabla f(\theta,x)\cdot \nabla\mathcal{V}(\theta,\mu_{t})\mu_{t}(\theta)\\
&=\int \nabla f(\theta,x)\cdot \sum_{j=1}^{n}\nabla_{\theta}f(\theta,x_{j})\left(y_{j}-\E_{\ov{\alpha}\sim \mu_{t}}\left[f(\ov{\alpha},x_{j})\right]\right)\mu_{t}(\theta)\\
&=\sum_{j=1}^{n}\int \nabla f(\theta,x)\cdot\nabla_{\theta}f(\theta,x_{j})\left(y_{j}-\E_{\ov{\alpha}\sim \mu_{t}}\left[f(\ov{\alpha},x_{j})\right]\right)\mu_{t}(\theta)\\
&=\sum_{j=1}^{n}\int \nabla f(\theta,x)\cdot\nabla_{\theta}f(\theta,x_{j})\left(y_{j}-f_{t}(x_{j})\right)\mu_{t}(\theta).
\end{align}

Let us define for $i,j=1,\ldots, n$ 

\begin{align}
K(x_{i},x_{j})\coloneqq \E_{\mu_{t}}\left[\nabla f(\theta,x_{i})\cdot\nabla_{\theta}f(\theta,x_{j})\right],
\end{align}
so that

\begin{align}
\frac{\de f_{t}(x_{i})}{\de t}=-\sum_{j=1}^{n}K(x_{i},x_{j})(f_{t}(x_{j})-y_{j}).    
\end{align}
This result shows that the dynamics of the limiting function $f_{t}$ is governed by a kernel-based learning algorithm where the time-dependent kernel $K(x_{i},x_{j})$ is defined by the expected gradient correlation over the measure $\mu_{t}$ \cite{abedi2023,girardi2025,melchor2025quantitative}. In what follows, the following bounds on the derivatives of the objective function $f(\theta,x)$ will be used

\begin{cor}\label{cor:derivatives}
Let $\theta=(\theta_{1},\ldots,\theta_{d})$. For all $j,k=1,\ldots,d$, one gets that
\begin{align}\label{girher1}
 &\vert \partial_{\theta_{k}}f(\theta,x) \vert\leq 1,\\\label{girher2}
 &\vert \partial_{\theta_{j}}\partial_{\theta_{k}}f(\theta,x)\vert\leq 1
\end{align}    
\end{cor}
These bounds are already obtained in the proof of \cite[Lemma 4.1]{hernandez2025meanfieldlimitgeneralmixtures}. Here, for the sake of completeness, we report their proof.
\begin{proof}
Let us set $W_{j}(\theta_{j})\coloneqq e^{-\frac{i\theta_{j}\mathcal{G}_{j}}{2}}$ for all $j=1,\ldots, d$, and $V_{k}\coloneqq V_{k}(x)$, for all $k=0,\ldots, d$. By the definition of $U$ as in \eqref{uformula}, we get that
\begin{align}
\partial_{\theta_{k}}f(\theta,x)&=\bra{0^{m}}V_{0}^{\dag}(x)W_{1}^{\dag}(\theta_{1})V_{1}^{\dag}\cdots\left(\partial_{\theta_{k}}W_{k}(\theta_{k})\right)^{\dag}V_{k}^{\dag}\cdots W_{d}^{\dag}(\theta_{d})V_{d}^{\dag}\O V_{d}W_{d}(\theta_{d})\cdots V_{1}W_{1}(\theta_{1})V_{0}\ket{0^{m}}\\
&+\bra{0^{m}}V_{0}^{\dag}(x)W_{1}^{\dag}(\theta_{1})V_{1}^{\dag}\cdots W_{d}^{\dag}(\theta_{d})V_{d}^{\dag}\O V_{d}W_{d}(\theta_{d})\cdots V_{k}\left(\partial_{\theta_{k}}W_{k}(\theta_{k})\right) \cdots V_{1}W_{1}(\theta_{1})V_{0}\ket{0^{m}}\\
&=\frac{i}{2}\bra{0^{m}}V_{0}^{\dag}(x)W_{1}^{\dag}(\theta_{1})V_{1}^{\dag}\cdots \mathcal{G}_{k}W_{k}^{\dag}(\theta_{k})V_{k}^{\dag}\cdots W_{d}^{\dag}(\theta_{d})V_{d}^{\dag}\O V_{d}W_{d}(\theta_{d})\cdots V_{1}W_{1}(\theta_{1})V_{0}\ket{0^{m}}\\
&-\frac{i}{2}\bra{0^{m}}V_{0}^{\dag}(x)W_{1}^{\dag}(\theta_{1})V_{1}^{\dag}\cdots W_{d}^{\dag}(\theta_{d})V_{d}^{\dag}\O V_{d}W_{d}(\theta_{d})\cdots V_{k}W_{k}(\theta_{k})\mathcal{G}_{k}\cdots V_{1}W_{1}(\theta_{1})V_{0}\ket{0^{m}}.
\end{align}
Since $\norma{V_{k}}\leq 1$, and $\norma{W_{j}}\leq 1$, we have that

\begin{align}
\vert\partial_{\theta_{k}}f(\theta,x)\vert\leq \norma{\mathcal{G}_{k}}\norma{\O}\leq 1. 
\end{align}
Let us now compute $\partial_{\theta_{j}}\partial_{\theta_{k}}f(\theta,x)$. Notice that

\begin{align}
\partial_{\theta_{j}}\partial_{\theta_{k}}f(\theta,x)&= \frac{i^{2}}{4}\bra{0^{m}}V_{0}^{\dag}(x)W_{1}^{\dag}(\theta_{1})V_{1}^{\dag}\cdots \mathcal{G}_{j}W_{j}^{\dag}(\theta_{j})V_{j}^{\dag}\cdots \\
&\cdots \mathcal{G}_{k}W_{k}^{\dag}(\theta_{k})V_{k}^{\dag}\cdots W_{d}^{\dag}(\theta_{d})V_{d}^{\dag}\O V_{d}W_{d}(\theta_{d})\cdots V_{1}W_{1}(\theta_{1})V_{0}\ket{0^{m}}\\
&+\frac{i(-i)}{4}\bra{0^{m}}V_{0}^{\dag}(x)W_{1}^{\dag}(\theta_{1})V_{1}^{\dag}\cdots\mathcal{G}_{k}W_{k}^{\dag}(\theta_{k})V_{k}^{\dag}\cdots W_{d}^{\dag}(\theta_{d})V_{d}^{\dag}\O V_{d}W_{d}(\theta_{d})\cdots\\
&\cdots V_{j}W_{j}(\theta_{j})\mathcal{G}_{j}\cdots V_{1}W_{1}(\theta_{1})V_{0}\ket{0^{m}}\\
&+ h.c..
\end{align}
Since $\norma{V_{k}}\leq 1$, and $\norma{W_{j}}\leq 1$, we then get that

\begin{align}
\vert \partial_{\theta_{j}}\partial_{\theta_{k}}f(\theta,x)\vert \leq \norma{\mathcal{G}_{j}}\norma{\mathcal{G}_{k}}\norma{\O}\leq 1. 
\end{align}
\end{proof}

\subsection{The functional setup}
In this subsection, we introduce the functional setting where our analysis will take place. In what follows, we make use of the so-called Poincaré-Wirtinger inequality.

\begin{thm}[{\cite[Proposition 3.38]{cioranescu}}]
Let $d\geq 1$, and let us consider the $d$-dimensional Torus $\T^{d}$. There exists a positive constant $C\coloneqq C(\mathbb{\T^{d}})$ depending on the volume of $\T^{d}$ such that
\begin{align}\label{wirtinger}
\norma{u-\M(u)}_{L^{2}(\T^{d})}\leq C\norma{\nabla u}_{L^{2}(\T^{d})},    
\end{align}
where $\M(u)\coloneqq \fint_{\T^{d}}u(x)\de x$.
\end{thm}
We have stated such an important inequality on $\T^{d}$ since our analysis concerns to quantum neural networks with parametrized on $\T^{d}$. However, this Theorem holds true for every connected open bounded set of class $C^{1}$ as sttaed in \cite[Chapter 9, page 312]{brezis2011functional}. Notice that if $u$ is constant, it satisfies \eqref{wirtinger}. Since we aim to study the convergence of a sequence of measures, and we aim to obtain a convergence result we then need to work with functions that equivalent up to a constant.

\begin{defn}
Let us denote by $\mathbb{W}^{1,2}(\T^{d})$ the Sobolev space of square integrable functions $\phi:\R^{d}\rightarrow \R$ with respect to the Lebesgue measure $\de x$  such that its gradient $\nabla \phi$ is square integrable on $\T^{d}$. More precisely,

\begin{align}
\mathbb{W}^{1,2}(\T^{d})\coloneqq \left\{\phi:\T^{d}\rightarrow \R: \phi\in L^{2}(\T^{d}),\, \nabla\phi\in L^{2}(\T^{d},\R^{d})\right\}.  
\end{align}
\end{defn}
It is customary to denote $\mathbb{H}^{1}(\T^{d})\coloneqq \mathbb{W}^{1,2}(\T^{d})$ because it is  Hilbert space when endowed with the norm

\begin{align}
\norma{\phi}_{\mathbb{H}^{1}}\coloneqq \left(\norma{\phi}_{L^{2}(\T^{d})}^{2}+ \norma{\nabla\phi}_{L^{2}(\T^{d})}^{2}\right)^{\frac{1}{2}}    
\end{align}
where we have used the notation

\begin{align}
\nabla\phi=\left(\frac{\partial \phi}{\partial x_{1}},\cdots,\frac{\partial \phi}{\partial x_{d}}\right),    
\end{align}
and 

\begin{align}
\norma{\nabla\phi}_{L^{2}(\T^{d})}\coloneqq \left(\sum_{i=1}^{d}\norma{\frac{\partial \phi}{\partial x_{i}}}_{L^{2}(\T^{d})}^{2}\right)^{\frac{1}{2}}.
\end{align}

\begin{defn}
We define the quotient 

\begin{align}
   \mathbb{W}(\T^{d})\coloneqq \mathbb{H}^{1}(\T^{d})/\R 
\end{align}
which is defined as the space of classes of equivalence with respect to the relation
\begin{align}
  \phi\sim \phi' \, \Longleftrightarrow \, \phi-\phi'\,\, \text{is a constant for all $\phi,\phi'\in \mathbb{H}^{1}(\T^{d})$.}
\end{align}
We denote by $[\phi]$ the class of equivalence represented by $\phi$.
\end{defn}
\begin{prop}[{\cite[Proposition 3.40]{cioranescu}}]
 The quantity 

 \begin{align}
 \norma{[\phi]}_{\mathbb{W}(\T^{d})}\coloneqq \norma{\nabla \phi}_{L^{2}(\T^{d})}\hskip 0,2cm \forall\, \phi\in [\phi],\, [\phi]\in \mathbb{W}(\T^{d})   
 \end{align}
 defines a norm on $\mathbb{W}(\T^{d})$ for which $\mathbb{W}(\T^{d})$ is a Banach space. Moreover, $\mathbb{W}(\T^{d})$ is a Hilbert space for the scalar product

 \begin{align}\label{scalarproduct}
   \left(\phi,\phi'\right)_{\mathbb{W}(\T^{d})}\coloneqq \sum_{i=1}^{d}\left(\frac{\partial \phi}{\partial x_{i}}\frac{\partial \phi'}{\partial x_{i}}\right)_{L^{2}(\T^{d})},
 \end{align}
 where

 \begin{align}
 \left(\frac{\partial \phi}{\partial x_{i}}\frac{\partial \phi'}{\partial x_{i}}\right)_{L^{2}(\T^{d})}\coloneqq \int_{\T^{d}}\frac{\partial \phi}{\partial x_{i}}\frac{\partial \phi'}{\partial x_{i}} \de x.
 \end{align}
\end{prop}
Another space which is of importance in the sequel is the dual space $(\mathbb{W}(\T^{d}))'$ of $\mathbb{W}(\T^{d})$, and which is endowed with the norm

\begin{align}\label{dualnorm}
\norma{\psi}_{(\mathbb{W}(\T^{d}))'}\coloneqq \sup_{[0]\neq [\phi]\in \mathbb{W}(\T^{d})}\frac{\vert \psi([\phi])\vert}{\norma{[\phi]}_{\mathbb{W}.(\T^{d})}}   
\end{align}
A further fundamental notion needed in the present work is that of weak topology, that we now briefly recall.

\begin{defn}[{\cite[See Section 3.2]{brezis2011functional}}]
The weak topology $\sigma(\mathbb{W}(\T^{d}),(\mathbb{W}(\T^{d}))')$ on $\mathbb{W}(\T^{d})$ is the coarsest topology on $\mathbb{W}(\T^{d})$ such that all the elements of the dual space $(\mathbb{W}(\T^{d}))'$ are continuous. If a sequence $\{[\phi_{k}]\}_{k\in\N}$ in $\mathbb{W}(\T^{d})$ converges to $\phi$ in the weak topology $\sigma(\mathbb{W}(\T^{d}),(\mathbb{W}(\T^{d}))')$ we shall write
\begin{align}
[\phi_{k}]	\rightharpoonup [\phi].
\end{align}
\end{defn}
An immediate consequence of the reflexivity property of $\mathbb{W}(\T^{d})$, and by Kakutani's theorem, is that the  unit ball in $\mathbb{W}(\T^{d})$ which is defined as

\begin{align}
    \mathbb{B}_{\mathbb{W}(\T^{d})}\coloneqq \left\{\phi\in \mathbb{W}(\T^{d}): \norma{\phi}_{\mathbb{W}(\T^{d})}\leq 1\right\}
\end{align}
is weakly compact ( see for instance \cite[Theorem 3.17]{brezis2011functional}). A further important tool that will be used is the following.

\begin{cor}[{\cite[Corollary 3.22]{brezis2011functional}}]\label{breziscompactness}
Let $K\subset \mathbb{W}(\T^{d})$ be a bounded, closed, and convex subset of $\mathbb{W}(\T^{d})$. Then $K$ is compact in the topology $\sigma(\mathbb{W}(\T^{d}),(\mathbb{W}(\T^{d}))')$.
\end{cor}
Here we have stated such an important result for our case of interest, but it is clear that this result holds true for any generic reflexive Banach space. On other hand, since in equation \eqref{dualnorm}, we are not longer involved with the constants, and since we aim to study the fluctuations of $\delta_{t}^{N}$, we then use the following set. Let us consider

\begin{align}\label{subset}
 S\coloneqq \left\{[\phi]:\text{$[\phi]\in \mathbb{W}(\T^{d})$, $\phi\in C_{b}(\R^{d})$, and $\norma{\nabla \phi}_{L^{\infty}}\leq 1$,\,$\norma{\frac{\partial^{2} \phi}{\partial x_{i}x_{j}}}_{L^{\infty}}\leq 1$ for $i,j=1,\ldots,d$ }\right\}. 
\end{align}
We then endow $S$ with the norm
\begin{align}\label{seminorm}
\norma{\phi}_{S}\coloneqq \displaystyle\norma{[\phi]}_{\mathbb{W}(\T^{d})}.
\end{align}
Notice that $S$ is not a linear space but it can be proved that $S$ is a bounded, closed and convex set with respect to the norm $\mathbb{W}(\T^{d})$.
\begin{lem}\label{fundamentallem}
 Let us define the subset $S$ of $\mathbb{W}(\T^{d})$ as in \eqref{subset}. Then $S$ is a bounded, closed, and convex set with respect to \eqref{seminorm}.
\end{lem}
\begin{proof}
 Let us first prove that $S$ is bounded. Indeed, notice that for any $[\phi]\in S$, one gets

\begin{align}
\norma{[\phi]}_{S}=\left(\displaystyle\int_{\T^{d}}\norma{\nabla \phi}^{2}\de x\right)^{\frac{1}{2}}\leq \left(\displaystyle\int_{\T^{d}}\norma{\nabla\phi}_{L^{\infty}}^{2}\de x\right)^{\frac{1}{2}}\leq [{\rm vol}(\T^{d})]^{\frac{1}{2}}  
\end{align}
where ${\rm vol}(\T^{d})$ denotes the volume of $\T^{d}$, and where we have used that $\norma{\nabla\phi}_{L^{\infty}}\leq 1$. Since the previous bound holds for every $[\phi]\in S$, then $S$ is bounded. Let us now suppose that we have a sequence $\{[\phi_{k}]\}_{k\in \N}\subset S$ such that $\norma{\nabla \phi_{k}}_{L^{\infty}}\leq 1$ for all $k$, and $\norma{[\phi_{k}]-[\phi]}_{\mathbb{W}(T^{d})}\rightarrow 0$ to some $[\phi]\in \mathbb{W}(\T^{d})$. We aim to show that $\phi\in S$. Notice that by the lower semicontinuity of the norm $\norma{\cdot}_{L^{\infty}}$, one has that

\begin{align}\label{semicontinuità}
\norma{\nabla \phi}_{L^{\infty}}\leq \liminf_{k\rightarrow +\infty}\norma{\nabla \phi_{k}}_{L^{\infty}}\leq 1.
\end{align}
Now since $\{\phi_{k}\}_{k\in \N}\subset C_{b}(\R^{d})$ this property is still preserved by $\phi$. In fact since \eqref{semicontinuità} holds true, by making use of the same argument of \cite[Remark 7, Equation 5]{brezis2011functional} one has that

\begin{align}\label{boundbrezis}
\vert \phi(x)-\phi(y) \vert\leq \norma{\nabla \phi}_{L^{\infty}}\vert x-y\vert,    
\end{align}
and since $\T^{d}$ is compact, we have that $\phi$ is bounded. It remains to prove the convexity. Let $\lambda\in [0,1]$ and define $\phi_{\lambda}\coloneqq \lambda\phi_{1}+(1-\lambda)\phi_{2}$, for $[\phi_{1}],[\phi_{2}]\in S$. We aim to prove that $[\phi_{\lambda}]\in S$. Notice that by definition, we only need to check that $\norma{\nabla\phi_{k}}_{L^{\infty}}\leq 1$. Indeed, one has that

\begin{align}
\norma{\nabla\phi_{\lambda}}_{L^{\infty}}& \leq \vert \lambda\vert\norma{\nabla\phi_{1}}_{L^{\infty}}+ \vert 1-\lambda\vert\norma{\nabla\phi_{2}}_{L^{\infty}}\\
&=\lambda\norma{\nabla\phi_{1}}_{L^{\infty}}+ 1-\lambda\norma{\nabla\phi_{2}}_{L^{\infty}}\\
&\leq 1
\end{align}
where we have used that $\lambda, 1-\lambda$ are non negative, and that $\norma{\nabla\phi_{i}}_{L^{\infty}}\leq $, for $i=1,2$. The same reasoning can be applied to $\frac{\partial ^{2}\phi}{\partial x_{i}\partial x_{j}}$ for all $i,j=1\ldots, d$. 
\end{proof}
In view of \autoref{fundamentallem}, and since $\mathbb{W}(\T^{d})$ is a Hilbert space, we can prove that $S$ is compact with respect to the weak convergence.
\begin{lem}
 Let us define $S$ as in \eqref{subset}. Then $S$ is compact with respect to the weak topology $\sigma(\mathbb{W}(\T^{d}),(\mathbb{W}(\T^{d}))')$. 
\end{lem}
\begin{proof}
This is a direct consequence of \autoref{fundamentallem} and \autoref{breziscompactness} by choosing $K=S$.
\end{proof}

\begin{rem}
The set $S$ defined in \eqref{subset} is a subset of the Sobolev space $\mathbb{W}(\T^{d})$. In particular we have that the linear span of $S$ contains the smooth periodic functions; hence $\overline{\operatorname{span}(S)}^{\mathbb{W}(\T^{d})}=\mathbb{W}(\T^d)$. Pairings $\nm{\cdot,\cdot}$ are understood in the $(\mathbb{W}(\T^d))'$--$\mathbb{W}(\T^d)$ sense. In the sequel we will often test distributional identities against functions $\phi\in S$ (for which we may use the pointwise bounds $\norma{\nabla\phi}_{L^{\infty}}\le 1$ and $\norma{\frac{\partial^2\phi}{\partial x_{i}\partial x_{j}}}_{L^{\infty}}\le 1$); when a bound is first proved for $\varphi\in S$, it extends to the whole unit ball of $\mathbb{W}^{1,2}(\T^{d})$ by density and continuity of the involved linear functionals. Strictly speaking, the set $S$ defined in~\eqref{subset} is only a convex, closed and bounded
subset of $\mathbb{W}(\T^d)$, and therefore it does not carry a linear structure nor a topological dual in the usual sense. 
To make the functional setting rigorous, we then set
\begin{align}
S'\coloneqq (\mathbb{W}(\T^{d}))'   
\end{align}
endowed with the dual norm
\begin{align}
\norma{\delta}_{S'}=\sup_{\|\varphi\|_{S}\le1} |\nm{\delta,\varphi}|.   
\end{align}
\end{rem}

\begin{lem}\label{lemma:density-extension}
Let $S$ and $S'$ be as above. Suppose that for every $\phi\in S$ and every $\delta\in S'$ we
have the estimate

\begin{align}
|\nm{\mathcal L(\delta),\phi}| \le C \norma{\delta}_{S'}\norma{\phi}_{\mathbb{W}(\T^{d})},   
\end{align}
for some linear operator $\mathcal{L}:S'\to S'$. Then the same estimate holds for all $\phi\in \mathbb{W}(\T^d)$, and consequently
\begin{align}
\norma{\mathcal L(\delta)}_{S'}\le C\norma{\delta}_{S'}\hskip 0,2cm \text{for all $\delta\in S'$.}   
\end{align}
\end{lem}
\begin{proof}
Since $\operatorname{span}(S)$ is dense in $\mathbb{W}(\T^{d})$, pick a sequence $\phi_n\in\operatorname{span}(S)$
with $\phi_n\to\phi$ in $\mathbb{W}$. By the assumed inequality the sequence $\nm{\mathcal{L}(\delta),\varphi_n}$ is Cauchy and bounded, hence converges to a limit which coincides with $\nm{\mathcal{L}(\delta),\phi}$ by continuity of $\mathcal{L}(\delta)\in S'$. Passing to the limit yields the claim.
\end{proof}

\section{Central Limit theorems for quantum neural networks}\label{sec:mainresults}
In the next, we are interested to study th asymptotic behavior of $\mu_{\Theta_{t}^{N}}$ around its limit $\mu_{t}$, and also to study its relation with those of $F_{N}(\Theta,x)$ towards $f_{t}(x)$ for any $t\in [0,T]$. Let us define the term 

\begin{align}
\delta_{t}^{N}\coloneqq \sqrt{N}(\mu_{\Theta_{t}^{N}}-\mu_{t}).
\end{align}
The scaling factor $\sqrt{N}$ is introduced because, according to the standard Central Limit Theorem, fluctuations around the mean typically decay at a rate of $N^{-\frac{1}{2}}$. The term $\delta_{t}^{N}$ can be expressed as a sum of centered and rescaled Dirac measures:

\begin{align}
\delta_{t}^{N}=N^{-\frac{1}{2}}\sum_{i=1}^{N}\left(\delta_{\theta_{t}^{i}}-\mu_{t}\right). 
\end{align}
The goal is to prove that this sequence of measures converges weakly in law to a limiting fluctuation process $\delta_{t}$. That is,

\begin{align}
\delta_{t}^{N}\rightarrow \delta_{t}\, \text{weakly as $N\rightarrow +\infty$.}    
\end{align}
To characterize the limit process, we make use of the continuity equation equation \eqref{FKP} for the empirical measure $\mu_{\Theta_{t}^{N}}$: 

\begin{align}\label{continuity_eq}
\frac{\de \mu_{\Theta_{t}^{N}}}{\de t}=-\nabla_{\theta}\cdot\left(\nabla \mathcal{V}(\theta,\mu_{\Theta_{t}^{N}})\mu_{\Theta_{t}^{N}}\right).   
\end{align}
By subtracting the continuity equation for the limit $\mu_{t}$, and multiplying by $\sqrt{N}$, the evolution of $\delta_{t}^{N}$ is derived. So that, 

\begin{align}
\frac{\de \delta_{t}^{N}}{\de t}&=-N^{\frac{1}{2}}\nabla_{\theta}\cdot \left(\mu_{\Theta_{t}^{N}}\nabla \mathcal{V}(\theta,\mu_{\Theta_{t}^{N}})-\mu_{t}\nabla\mathcal{V}(\theta,\mu_{t})\right)\\
&=-N^{\frac{1}{2}}\nabla_{\theta}\cdot \left(\left(\mu_{\Theta_{t}^{N}}-\mu_{t}\right)\nabla \mathcal{V}(\theta,\mu_{t})+\mu_{\Theta_{t}^{N}}\left(\nabla\mathcal{V}(\theta,\mu_{\Theta_{t}^{N}})-\nabla\mathcal{V}(\theta,\mu_{t})\right)\right).
\end{align}

The core of the derivation involves linearizing the flux term $\mu\nabla\mathcal{V}(\theta,\mu)$ where we have that

\begin{align}
\nabla\mathcal{V}(\theta,\mu_{\Theta_{t}^{N}})-\nabla\mathcal{V}(\theta,\mu_{t})= \int \nabla W(\theta,\theta')\left(\mu_{\Theta_{t}^{N}}-\mu_{t}\right)(\de \theta'),
\end{align}
and that $\mu_{\Theta_{t}^{N}}=\mu_{t}+N^{-\frac{1}{2}}\delta_{t}^{N}$. Let us define

\begin{align}
\mathcal{G}(\theta,\mu)\coloneqq -\E_{\theta'\sim \mu}[W(\theta,\theta')].    
\end{align}
Then we can express the continuity equation for $\delta_{t}^{N}$ as follows:

\begin{align}
\begin{aligned}\label{limiteq}
\frac{\de \delta_{t}^{N}}{\de t}&=-\nabla_{\theta}\cdot \left(\delta_{t}^{N}\nabla \mathcal{V}(\theta,\mu_{t})+\mu_{\Theta_{t}^{N}}\nabla\mathcal{G}(\theta,\delta_{t}^{N})\right)\\
&-\nabla_{\theta}\cdot \left(\delta_{t}^{N}\nabla \mathcal{V}(\theta,\mu_{t})+(\mu_{t}+N^{-\frac{1}{2}}\delta_{t}^{N})\nabla\mathcal{G}(\theta,\delta_{t}^{N})\right).
\end{aligned}
\end{align}
The term $N^{-\frac{1}{2}}\delta_{t}^{N}\nabla\mathcal{G}(\theta,\delta_{t}^{N})$ is a higher-order term that vanishes as $N\rightarrow +\infty$. The resulting equation for the limit $\delta_{t}$ will be a linear, deterministic evolution equation  on the space of signed measures. The analysis begins by establishing the convergence of the fluctuation at the initial time $t=0$. The existence of this limit is crucial for proving the Central Limit Theorem for the dynamics of the system, since it suggests that this property of convergence propagates in time when the parameters follow a gradient flow equation.
\subsection{Our main results}
This subsection presents the two key results that constitute the Central Limit Theorem for the system, demonstrating how the fluctuations around the mean-field limit evolve over time.
\begin{prop}\label{prop:initialcase}
Suppose that $\Theta_{0}^{N}$ is composed by independent, and identically distributed random variables, and that $\theta_{0}^{1}$ has law ${\rm unif}(\T^{d})$. Let $f$ be a quantum neural network as defined in \eqref{model1}. Then 

\begin{align}\label{compatcness}
\delta_{0}^{N}\rightarrow \mathcal{N}(0,\Sigma_{\ast})\, \hskip 0,2cm \text{in law as $N\rightarrow +\infty$.}
\end{align}

where $\mathcal{N}(0,\Sigma_{\ast})$ is a Gaussian variable with zero mean and variance $\Sigma_{\ast}$ given by

\begin{align}
\Sigma_{\ast}(\phi)\coloneqq\E\left[\phi(\theta_{0}^{1})-\E[\phi(\theta_{0}^{1})]\right]^{2}, \hskip0,2cm \phi\in C_{b}(\T^{d}).
\end{align}
\end{prop}
\autoref{prop:initialcase} establishes the nature of the fluctuations at the initial time $t=0$, before the gradient flow dynamics begin. This result states that the initial fluctuations of the empirical measure $\delta_{0}^{N}$ converge in law to a Gaussian process. It confirms that even though the particles (experts) start with arbitrary, independent distributions, the collective fluctuation term is governed by Gaussian statistics.

\begin{thm}\label{mainprop1}
Let $f$ be a quantum neural network as defined in \eqref{model1}. Then 

\begin{align}\label{compatcness}
\delta_{t}^{N}\rightarrow \delta_{t}\, \hskip 0,2cm \text{in law as $N\rightarrow +\infty$}
\end{align}
where 

\begin{align}\label{limitFKP}
\frac{\de \delta_{t}(\theta)}{\de t}=-\nabla_{\theta}\cdot\left(\nabla\mathcal{V}(\theta,\mu_{t})\delta_{t}+\mu_{t}\nabla\mathcal{G}(\theta,\delta_{t})\right).
\end{align}
\end{thm}
\autoref{mainprop1} is the main theorem, extending the initial convergence to the entire time interval $[0,T]$. This result confirms that the convergence established at $t=0$ propagates in time: the fluctuation process $\delta_t^N$ converges in law to a limiting process $\delta_t$ for all $t \geq 0$. The limiting fluctuation process $\delta_t$ is the unique solution to the following linearized transport equation \eqref{limitFKP}. Unlike the nonlinear continuity equation \eqref{FKP} governing $\mu_t$, the equation for $\delta_t$ is linear. This is a simplification happens once the system reaches the mean-field limit $\mu_t$, and thus the fluctuations around that limit evolve linearly. The equation for $\delta_t$ is a deterministic partial differential equation (PDE) on the space of signed measures. Although, in our analysis, we will show that $\delta_{t}$ belongs to the dual space $S'$ as defined in \eqref{subset} above. An immediate consequnce of \autoref{mainprop1} is the following:

\begin{cor}\label{convergenzadif}
Let us consider the same setting of \autoref{mainprop1}. Let us define for each $t\in[0,T]$ the random process

\begin{align}
\psi_{N}(t,x)\coloneqq \sqrt{N}(F_{N}(\Theta_{t},x)-f_{t}(x)).    
\end{align}
Then for any $x\in \R^{d}$, and $t\in [0,T]$

\begin{align}
\psi_{N}(t,x)\rightarrow \nm{f_{t}(x),\delta_{t}}\, \hskip 0,2cm \text{in law as $N\rightarrow +\infty$.} 
\end{align}
\end{cor}
Let us now state our large deviation result. To this aim, let us use the following notation. Fix $T>0$ arbitrary. We denote by $C_{T}\coloneqq C([0,T];\T^{d})$, $\C_{T}\coloneqq C([0,T];\cP(\T^{d}))$ the spaces of continuous maps from $[0,T]$ to $\T^{d}$, and $\cP(\T^{d})$, respectively. Let $\mu_{0}\in\cP(\T^d)$ be the law of the i.i.d. initial data $\{\Theta_{0}^{i,N}\}_{i=1}^{N}$. Let us take the sequence $\ov\Theta^{N}$ from propagation of chaos such that for $d>4$ 

\begin{align}
\E\W_{2}^{2}(\mu_{\Theta_{t}^{N}},\mu_{\ov\Theta_{t}^{N}})\leq C_{1}N^{-\frac{2}{d}}    
\end{align}
for some positive constant $C_{1}$ independent of $N$. Let us consider the probability measure $\mu_{0}^{\otimes N}$, and let $\mathbb P_{\mu_{0}}^{N}$ be the probability measure on $\C_{T}$ given by the image of $\mu_{0}^{\otimes N}$ under the map

\begin{align}\label{map_used}
(C_{T})^{N}\ni\ov\Theta^{N}  =\bigl(\ov\Theta^{1,N}_t,\dots,\ov\Theta^{N,N}_t\bigr)_{t\in[0,T]}\longmapsto \left(t\longmapsto \frac{1}{N}\sum_{i=1}^{N}\delta_{\ov\Theta^{i,N}_t}\right)\in \C_{T}. 
\end{align}
The following result will be a consequence of \autoref{thm:gen_sanov}. In what follows, we set on $\C_{T}$
\begin{align}\label{metric_X}
\de_{T,\cP(\T^{d})}(\y_{1},\y_{2})\coloneqq\sup_{t\in [0,T]}\W_{2}(\y_{1}(t),\y_{2}(t)), \hskip 0,2cm \text{for all $\y_{1},\y_{2}\in \C_{T}$,}
\end{align}
where $\W_{2}$ is the distance of Wasserstein of order $2$.
\begin{thm}\label{thm:LDP}
Let us fix $T>0$, $0<\beta<1$, and $d>4$. Let $\mu_{0}\in\cP(\mathbb T^d)$, and denote by $\mu_{\Theta_{\cdot}^{N}}:t\in[0,T]\mapsto \mu_{\Theta_{t}^{N}}$, and $\mu_{\ov\Theta_{\cdot}^{N}}:t\in[0,T]\mapsto \mu_{\ov\Theta_{t}^{N}}$. Then the following assertions hold true:

\begin{itemize}
    \item[a.]$(\C_{T},\P_{\mu_{\Theta_{\cdot}^{N}}}^{N},N)$, and $(\C_{T},\P_{\mu_{\ov\Theta_{\cdot}^{N}}}^{N},N)$ satisfy a large deviation principle  both with  the same rate functional $L_{\mu_{\cdot}}(\nu)$ where $\nu\in \C_{T}$, and $\mu_{\cdot}$ the unique solution to \eqref{continuity_eq}.
    \item[b.] Let us suppose that $\Theta_{0}^{N}$ is composed by independent and identically distributed random variables with values in $\mathbb{T}^{d}$. Then $\mu_{\ov\Theta_{\cdot}^{N}}^{N}$, and $\mu_{\Theta_{\cdot}^{N}}^{N}$ are exponentially equivalent with speed $(\log N)^{\beta}$. That is, for any $\varepsilon>0$, one gets

\begin{align}\label{equiv}
\lim_{N\to\infty}\frac{1}{(\log N)^{\beta}}\log\mathbb P\Big(\de_{T,\cP(\T^{d})}(\mu_{\Theta_{\cdot}^{N}},\mu_{\ov\Theta_{\cdot}^{N}})>\varepsilon\Big)
=-\infty.
\end{align}
\end{itemize}
\end{thm}
\begin{rem}
Let us comment on the content of the previous theorem. Observe that, while the random variables appearing in $\mu_{\Theta_{\cdot}^{N}}$ may be correlated, the propagation of chaos allows us to approximate $\Theta_{\cdot}^{N}$ by independent random variables. The price to pay is that the convergence rate in the Wasserstein distance of order $2$ can be established only for $d>4$, with rate $N^{-2/d}$ as $N\to\infty$. This prevents us from obtaining a classical exponential equivalence between these processes at speed $N$, as usually stated in \cite[Definition 4.2.10]{dembo2011}.
\end{rem}
\section{Proofs}\label{proofs}
In this section, we prove our main results \autoref{prop:initialcase}, and \autoref{mainprop1}.
\begin{proof}[{Proof of \autoref{prop:initialcase}}]
Let us take any $\phi\in C_{b}(\T^{d})$. Since for $t=0$, $\Theta_{0}$ is composed by independent and idnticcaly distributed random variables, then by the central limit theorem

\begin{align}
\displaystyle\int_{\T^{d}}\phi(\theta)\delta_{0}^{N}=N^{-\frac{1}{2}}\sum_{i=1}^{N}\left(\phi(\theta_{0}^{i})-\displaystyle\int_{\T^{d}}\phi(\theta)\mu_{0}(\de \theta)\right)\rightarrow \mathcal{N}(0,\Sigma_{\phi})   
\end{align}
where $\mathcal{N}(0,\Sigma_{\phi})$ is a Gaussian variable with zero mean and variance $\Sigma_{\phi}$ given by

\begin{align}
\Sigma_{\phi}=\E\left[\phi(\theta_{0}^{1})-\E[\phi(\theta_{0}^{1})]\right]^{2}.  
\end{align}
\end{proof}
In what follows, we prove the relative compactness of $\{\delta^{N}\}_{N\in\N}$. Let us first prove the following:

\begin{lem}\label{auxlem1}
For all $t\in [0,T]$, one gets that 
\begin{align}\label{toestimate}
\sup_{N\in\N}\E\left[\norma{\delta_{t}^{N}}_{S'}^{2}\right]<+\infty.  
\end{align}   
\end{lem}
\begin{proof}
We first show that this property holds true for $t=0$, and that this propagates for all $0<t\leq T$. 
Let denote by $\nm{\psi,\phi}$ the duality bewteen $\psi\in S'$, and $\phi\in S$ where we identify $\phi$ with its class of equivalence $[\phi]$. Notice that
\begin{align}
\nm{\delta_{0}^{N},\phi}=\sqrt{N}\left(\frac{1}{N}\sum_{i=1}^{N}\left(\phi(\theta_{t}^{i})-\nm{\mu_{0},\phi}\right)\right),   
\end{align}
and thus
\begin{align}
\sup_{[0]\neq [\phi]\in S}\nm{\delta_{0}^{N},\phi}\leq \frac{1}{\sqrt{N}}\sum_{i=1}^{N}\sup_{[0]\neq [\phi]\in S}\left\vert \phi(\theta_{t}^{i})-\mu_{0}(\phi) \right\vert.  
\end{align}
Since all $(\theta_{0}^{i})_{i=1}^{N}$ are independent, and identically distributed, then one gets that

\begin{align}
\sup_{N\in\N}\E\left(\sup_{[0]\neq [\phi]\in S}\frac{\nm{\delta_{0}^{N},\phi}}{\norma{[\phi]}_{S}}\right)^{2}&\leq \E\left(\sup_{[0]\neq [\phi]\in S}\frac{\left\vert \phi(\theta_{0}^{1})-\nm{\mu_{0},\phi} \right\vert}{\norma{[\phi]}_{S}}\right)^{2}\\
\end{align}
where $\mu_{0}={\rm unif}(\T^{d})$. Since 

\begin{align}
\E\left(\sup_{[0]\neq [\phi]\in S}\frac{\left\vert \phi(\theta_{0}^{1})-\nm{\mu_{0},\phi} \right\vert}{\norma{[\phi]}_{S}}\right)^{2}=\E\left(\sup_{\norma{[\phi]}_{S}=1}\left\vert \phi(\theta_{0}^{1})-\nm{\mu_{0},\phi} \right\vert\right)^{2}    
\end{align}
then we get

\begin{align}
\E\left(\sup_{\norma{[\phi]}_{S}=1}\left\vert \phi(\theta_{0}^{1})-\nm{\mu_{0},\phi} \right\vert\right)^{2}    = 
\displaystyle\frac{1}{{\rm vol}(\T^{d})}\int_{\T^{d}}\left(\sup_{\norma{[\phi]}_{S}=1}
\left\vert \phi(\theta)-\nm{\mu_{0},\phi} \right\vert\right)^{2}\de \theta.
\end{align}
On the other hand, since $\norma{\nabla \phi}_{L^{\infty}}\leq 1$, by \eqref{boundbrezis} we get

\begin{align}
\left\vert \phi(\theta)-\nm{\mu_{0},\phi} \right\vert\leq \frac{1}{{\rm vol}(\T^{d})}\displaystyle\int_{\T^{d}}\vert \phi(\theta)-\phi(x)\vert\de x&\leq  \frac{1}{{\rm vol}(\T^{d})}\displaystyle\int_{\T^{d}}\vert\theta-x\vert \de x\\
&\leq \sup_{x\in \T^{d}}\vert \theta-x\vert\\
&={\rm diam}(\T^{d})
\end{align}
where ${\rm diam}(\T^{d})$ denotes the diameter of $\T^{d}$. Therefore,

\begin{align}
\sup_{N\in\N}\E\left(\sup_{[0]\neq [\phi]\in S}\frac{\nm{\delta_{0}^{N},\phi}}{\norma{[\phi]}_{S}}\right)^{2}\leq \left({\rm diam}(\T^{d})\right)^{2}.    
\end{align}
Let us now observe that \eqref{toestimate} also holds true for $0<t\leq T$. Let us take $[\phi]\in S$ such that $\norma{[\phi]}_{S}=1$. Since for all $a,b\in \R$, $2ab\leq a^{2}+b^{2}$, in combination with Jensen inequality one gets

\begin{align}
\left(\sup_{\norma{[\phi]}_{S}=1}\nm{\delta_{t}^{N},\phi}\right)^{2}&\leq 2\left(\sup_{\norma{[\phi]}_{S}=1}\vert\nm{\delta_{0}^{N},\phi}\vert\right)^{2}+\\ &+2t\int_{0}^{t}\left(\sup_{\norma{[\phi]}_{S}=1}\vert\nm{\delta_{s}^{N}\nabla \mathcal{V}(\theta,\mu_{s})+(\mu_{s}+N^{-\frac{1}{2}}\delta_{s}^{N})\nabla\mathcal{G}(\theta,\delta_{s}^{N}),\nabla\phi}\vert\right)^{2}\de s
\end{align}
where we recall that

\begin{align}
&\mathcal{V}(\theta,\mu)=V(\theta)- \E_{\theta'\sim \mu}[W(\theta,\theta')],\\
&V(\theta)\coloneqq  \sum_{j=1}^{n}f(\theta,x_{j})y_{j},\\
&W(\theta,\theta')\coloneqq \sum_{j=1}^{n}f(\theta,x_{j})f(\theta',x_{j}).
\end{align}
Since by \eqref{model1}, one has that $\vert \partial_{\theta_{i}}f(\theta,x)\vert\leq 1$ for all $i=1,\ldots,d$, then

\begin{align}\label{unifb1}
\norma{\nabla\mathcal{V}(\cdot,\mu)}_{L^{\infty}}\leq 2n 
\end{align}
where we have used also that $\vert f(\theta,x)\vert\leq 1$, and $\vert y_{j} \vert\leq 1$ for all $\theta,x$, and all $j=1,\ldots,d$. Furthermore, using the same reasoning one gets that 
\begin{align}\label{uniformb}
\norma{\nabla \mathcal{G}}_{L^{\infty}}\leq n.  
\end{align}
Then 

\begin{align}
\begin{aligned}\label{used1}
\vert\nm{\delta_{s}^{N},\nabla\mathcal{V}(\theta,\mu_{s})\cdot\nabla\phi}\vert&\leq \norma{\delta_{s}^{N}}_{S'}\norma{\nabla\mathcal{V}(\theta,\mu_{s})\cdot\nabla\phi}_{S}.
\end{aligned}
\end{align}
Let us notice that 

\begin{align}
\norma{\nabla\mathcal{V}(\theta,\mu_{s})\cdot\nabla\phi}_{S}^{2}\leq \int_{\T^{d}}\norma{\nabla\left(\nabla\mathcal{V}(\theta,\mu_{s})\cdot\nabla\phi\right)}_{L^{\infty}}^{2}\de \theta 
\end{align}
Notice that 

\begin{align}
\norma{\nabla(\nabla\mathcal{V}(\theta,\mu_{s})\cdot\nabla\phi)}_{2}\leq \norma{D_{\theta}^{2}\mathcal{V}(\theta,\mu_{s})}_{{\rm op}}\norma{\nabla\phi}_{2} + \norma{D_{\theta}^{2}\phi}_{{\rm op}}\norma{\nabla\mathcal{V}(\theta,\mu_{s})}_{2},
\end{align}
where $D_{\theta}^{2}$ denotes the Hessian with respect $\theta$, and $\norma{\cdot}_{{\rm op}}$ the operator norm. Since $\norma{\cdot}_{{\rm op}}\leq \norma{\cdot}_{F}$ (where $F$ stands for the Frobenius norm) we have that 

\begin{align}
\norma{D_{\theta}^{2}\phi}_{{\rm op}}\leq \norma{D_{\theta}^{2}\phi}_{F}=\sqrt{\sum_{i,j=1}^{d}\left(\frac{\partial^{2} \phi}{\partial x_{i}\partial x_{j}}\right)^{2}}\leq d.    
\end{align}
Furthermore

\begin{align}
\norma{\nabla\mathcal{V}(\theta,\mu_{s})}_{2}^{2}\leq \displaystyle\int_{\T^{d}}\norma{\nabla\mathcal{V}(\theta,\mu_{s})}_{L^{\infty}}^{2}\leq (2n)^{2} {\rm vol}(\T^{d})   
\end{align}
where ${\rm vol}(\T^{d})$ denotes the volume of $\T^{d}$. Furthermore, since by \autoref{cor:derivatives} $\left\vert\frac{\partial^{2} f(\theta)}{\partial x_{i}\partial x_{j}}\right\vert \leq 1$, we get also get that 

\begin{align}
\norma{D_{\theta}^{2}\mathcal{V}}_{{\rm op}}\leq d.   
\end{align}
Therefore,

\begin{align}\label{unifb2}
\begin{aligned}
\norma{\nabla(\nabla\mathcal{V}(\theta,\mu_{s})\cdot\nabla\phi)}_{2}&\leq d\sqrt{{\rm vol}(\T^{d})}+ 2nd \sqrt{{\rm vol}(\T^{d})}\\
&\leq 3nd \sqrt{{\rm vol}(\T^{d})}.
\end{aligned}
\end{align}

On the other hand, we have by \eqref{uniformb} that

\begin{align}\label{used2}
\begin{aligned}
\left\vert\nm{\left(\mu_{s}+N^{-\frac{1}{2}}\delta_{s}^{N}\right)\nabla\mathcal{G}(\theta,\delta_{s}^{N}),\nabla\phi}\right\vert &\leq \left\vert\nm{\mu_{s},\nabla\mathcal{G}(\theta,\delta_{s}^{N})\cdot\nabla \phi} \right\vert + \left\vert\nm{N^{-\frac{1}{2}}\delta_{s}^{N},\nabla\mathcal{G}(\theta,\delta_{s}^{N})\cdot\nabla \phi} \right\vert \\
&\leq nd +\norma{\delta_{s}^{N}}_{S'}\norma{\nabla\mathcal{G}(\theta,\delta_{s}^{N})\cdot\nabla \phi}_{S} \\
&\leq nd + 3nd \sqrt{{\rm vol}(\T^{d})}\norma{\delta_{s}^{N}}_{S'}
\end{aligned}
\end{align}
Therefore,

\begin{align}\label{auxbound1}
\begin{aligned}
\E\left[\norma{\delta_{t}^{N}}_{S'}^{2}\right]&\leq 2\E\left[\norma{\delta_{0}^{N}}_{S'}^{2}\right]+ 2t\E\int_{0}^{t}(nd + 6nd \sqrt{{\rm vol}(\T^{d})}\norma{\delta_{s}^{N}}_{S'})^{2}\de s\\
&\leq  2\E\left[\norma{\delta_{0}^{N}}_{S'}^{2}\right]+4(ndt)^{2}+ 144t(nd)^{2}{\rm vol}(\T^{d})\int_{0}^{t}\E\norma{\delta_{s}^{N}}_{S'}^{2}\de s.
\end{aligned}
\end{align}
Then by Gr\"{o}nwall inequality we get

\begin{align}
\E\!\left[\|\delta_t^N\|_{S'}^2\right]
&\le 
\Big(1+144 (n d)^2{\rm vol}(\T^d)\,t^2\Big)
\exp\!\Big(72 (n d)^2{\rm vol}(\T^d)\,t^2\Big)\notag\\
&\qquad\times\Big(2\,\E\!\left[\|\delta_0^N\|_{S'}^2\right]+4 (n d)^2 t^2\Big).
\end{align}
Since this bound is uniform on $N$, then we conclude that \eqref{toestimate} holds true. 
\end{proof}
By following the previous reasoning, it is also possible to prove the following bound:

\begin{lem}\label{auxlem2}
\begin{align}\label{strongcadlag}
   \sup_{N\in\N}\E\left[\sup_{0\leq t\leq T}\nm{\delta_{t}^{N},\phi}^{2}\right]&\leq \Big(1+144 (n d)^2{\rm vol}(\T^d)\,T^2\Big)
\exp\!\Big(72 (n d)^2{\rm vol}(\T^d)\,T^2\Big)\notag\\
&\qquad\times\Big(2\,\E\!\left[\|\delta_0^N\|_{S'}^2\right]+4 (n d)^2 T^2\Big).
\end{align}        
\end{lem}
\begin{proof}
Notice that 
\begin{align}
\nm{\delta_{t}^{N},\phi}^{2}&\leq 2\vert\nm{\delta_{0}^{N},\phi}^{2}+\\ &+2t\int_{0}^{t}\left\vert\nm{\delta_{s}^{N}\nabla \mathcal{V}(\theta,\mu_{s})+(\mu_{s}+N^{-\frac{1}{2}}\delta_{s}^{N})\nabla\mathcal{G}(\theta,\delta_{s}^{N}),\nabla\phi}\right\vert^{2}\de s
\end{align}
and thus by \eqref{used1}, and \eqref{used2} one gets the desired bound \eqref{strongcadlag}.
\end{proof}
With this bound at hand, we may prove that the trajectories of $\delta_{t}^{N}$ are almost surely continuous in $S'$. We have the following
\begin{lem}\label{auxlem3}
We have that the trajectories of $t\mapsto\delta_{t}^{N}$ are almost surely continuous in $S'$.
\end{lem}
\begin{proof}
To prove this Lemma, we closely follows \cite[Proposition 3.5]{ferland1992compactness}. Notice that by following the same argument of \autoref{auxlem2} and since $\mathbb{W}(\T^{d})$ is a separable Hilbert space, we can choose an orthonormal basis $\{\phi_{k}\}_{k\in\N}$ of  $\mathbb{W}(\T^{d})$, and obtain that

\begin{align}\label{cadlag1}
   \sup_{N\in\N}\sum_{k\geq 1}\E\left[\sup_{0\leq t\leq T}\nm{\delta_{t}^{N},\phi_{k}}^{2}\right]&\leq \Big(1+144 (n d)^2{\rm vol}(\T^d)\,T^2\Big)
\exp\!\Big(72 (n d)^2{\rm vol}(\T^d)\,T^2\Big)\notag\\
&\qquad\times\Big(2\,\E\!\left[\|\delta_0^N\|_{S'}^2\right]+4 (n d)^2 T^2\Big).
\end{align} 
This bound implies that there exists  $\Omega_{0}\subset \Omega$ ($\Omega$ represents the ambient probability space) with $\P(\Omega_{0})=1$ such that for all $\omega\in \Omega_{0}$ one has

\begin{align}\label{utile1}
\sum_{k\geq 1}\sup_{0\leq t\leq T}\nm{\delta_{t}^{N},\phi_{k}}^{2}<+\infty.
\end{align}
In what follows, we prove that $t\mapsto \delta_{t}^{N}(\omega)$ is right-continuous. Let us fix $\omega\in \Omega_{0}$, and consider $t_{n}\downarrow t$ in $[0,T]$. We want to prove that for any $\varepsilon>0$, there exists $n(\varepsilon)\in\N$ such that for all $n\geq n(\varepsilon)$ one has $\norma{\delta_{t_{n}}^{N}(\omega)-\delta_{t}^{N}(\omega)}_{S'}<\varepsilon$. Indeed by the Parseval identity one has that

\begin{align}
\norma{\delta_{t_{n}}^{N}(\omega)-\delta_{t}^{N}(\omega)}_{S'}^{2}=\sum_{k=1}^{+\infty}\nm{\delta_{t_{n}}^{N}(\omega)-\delta_{t}^{N}(\omega),\phi_{k}}^{2}.   
\end{align}
By \eqref{utile1}, we then choose $M$ large enough such that

\begin{align}
\sum_{k>M}\sup_{0\leq t\leq T}\nm{\delta_{t}^{N},\phi_{k}}^{2}<\frac{\varepsilon^2}{6}.
\end{align}
Then,

\begin{align}\label{pbound1}
\begin{aligned}
\norma{\delta_{t_{n}}^{N}(\omega)-\delta_{t}^{N}(\omega)}_{S'}^{2}&=\sum_{k=1}^{+\infty}\nm{\delta_{t_{n}}^{N}(\omega)-\delta_{t}^{N}(\omega),\phi_{k}}^{2}\\
&\leq \sum_{k=1}^{M}\nm{\delta_{t_{n}}^{N}(\omega)-\delta_{t}^{N}(\omega),\phi_{k}}^{2} + 2\left\{\sum_{k>M}\nm{\delta_{t_{m}}^{N},\phi_{k}}^{2}+\nm{\delta_{t}^{N},\phi_{k}}^{2}\right\}.
\end{aligned} 
\end{align}
Now since the process $t\mapsto\nm{\delta_{t}^{N},\phi_{k}}$ is right continuous, we can take some $n(k,\varepsilon)$ such that for all $n\geq n(k,\varepsilon)$ one has 

\begin{align}
 \vert\nm{\delta_{t_{n}},\phi_{k}}-\nm{\delta_{t},\phi_{k}}\vert^{2}\leq \frac{\varepsilon^2}{3M}.  
\end{align}
Therefore, by \eqref{pbound1} we conclude that

\begin{align}
\begin{aligned}
\norma{\delta_{t_{n}}^{N}(\omega)-\delta_{t}^{N}(\omega)}_{S'}^{2}
&\leq \sum_{k=1}^{M}\frac{\varepsilon^{2}}{3M}+\frac{2\varepsilon^{2}}{3}\\
&=\varepsilon^{2}.
\end{aligned} 
\end{align}
Since we can use the reasoning to prove that $t\mapsto \delta_{t}^{N}(\omega)$ is left-continuous, we are done.
\end{proof}
Let us also notice that by \autoref{auxlem2}, and by Markov inequality we have that the sets 

\begin{align}
K_{n}\coloneqq \{\psi\in S': \norma{\psi}_{S'}^{2}\leq 2^{n}C\}    
\end{align}
where
\begin{align}
C\coloneqq \sup_{N\in\N}\E\left[\sup_{0\leq t\leq T}\norma{\delta_{t}^{N}}^{2}\right]<+\infty
\end{align}
satisfy that for all $n,N\in\N$ 

\begin{align}\label{decayprob}
\P\left\{\exists\,t\in [0,T]:\delta_{t}^{N}\notin K_{n}\right\}\leq 2^{-n}.  
\end{align}
Furthermore, by Markov inequality, we can prove the following:
\begin{lem}
\begin{align}\label{auxlem4}
\lim_{M\rightarrow +\infty}\sup_{N\in\N}\P\left[\left\{\sup_{0\leq t\leq T}\vert\nm{\delta_{t}^{N},\phi}\vert>M\right\}\right]=0.
\end{align}
\end{lem}
\begin{proof}
Recall that 
\begin{align}
\sup_{0\le t\le T} |\nm{\delta_t^N, \phi}|
\le \norma{\phi}_{S} \sup_{0\le t\le T} \norma{\delta_{t}^{N}}_{S'}.
\end{align}
Applying Markov’s inequality to the random variable $\big(\sup_{t}|\langle \delta_t^N, \phi\rangle|\big)^2$, we obtain

\begin{align}
\P\!\left(\sup_{0\le t\le T} |\nm{\delta_t^N, \phi}| > M\right)
&= \P\!\left(\big(\sup_{0\le t\le T}|\nm{\delta_t^N, \phi}|\big)^2 > M^2\right) \\
&\le \frac{\E\!\left[\big(\sup_{0\le t\le T}|\nm{ \delta_t^N, \phi}|\big)^2\right]}{M^2}.
\end{align}
Moreover,
\begin{align}
\E\!\left[\left(\sup_{0\leq t\leq T}|\langle \delta_t^N, \phi\rangle|\right)^2\right]
\le \norma{\phi}_{S}^2\, \E\!\left[\sup_{t}\norma{\delta_t^N}_{S'}^2\right]
\le \|\phi\|_{S}^2\, C.
\end{align}
Hence,
\begin{align}
\P\!\left(\sup_{0\le t\le T} |\nm{\delta_t^N, \phi}| > M\right) \le \frac{\norma{\phi}_{S}^2\,C}{M^2}, \qquad \forall\, N\in\mathbb{N}.
\end{align}
Taking the supremum over $N$ and then the limit as $M\to +\infty$, we obtain
\begin{align}
\lim_{M\to +\infty}\sup_{N\in\mathbb{N}}
\P\!\left(\sup_{0\le t\le T} |\nm{\delta_t^N, \phi}| > M\right)= 0.
\end{align} 
\end{proof}
Notice that if we prove that for all $\phi\in S$, the process $\{\nm{\delta^{N},\phi}\}_{N\in\N}$ is relatively compact together to \eqref{decayprob}, we obtain that $\delta^{N}$ is relatively compact.

\begin{lem}\label{auxlem5}
The following holds true: For all $\phi\in S$, for all $\varepsilon>0$ there exists $\delta>0$ and $M>0$ such that

\begin{align}
 \sup_{N>M}\P\left[\left\{\sup_{\stackrel{s,t\in [0,T]}{\vert t-s\vert<\delta}}\left\vert \nm{\delta_{t}^{N},\phi}-\nm{\delta_{s}^{N},\phi}\right\vert\geq \varepsilon\right\}\right]\leq \varepsilon. 
\end{align}
\end{lem}
\begin{proof}
Fix $\phi\in S$. By the weak formulation we have the integral representation

\begin{align}
\nm{\delta^N_t,\phi}-\nm{\delta^N_s,\phi}=\int_{s}^{t}A_{u}^{N}(\phi)\,\de u,    
\end{align}
where

\begin{align}
A_{u}^{N}(\phi)
=\nm{\delta_{u}^{N},\nabla_\theta \mathcal{V}(\cdot,\mu_u)\cdot\nabla\phi}
+\nm{\mu_{u}+N^{-1/2}\delta^N_u,\nabla_\theta \mathcal{G}(\cdot,\delta^N_u)\cdot\nabla\varphi}.
\end{align}
By the estimates already derived in \autoref{auxlem1}, there exists a constant $C>0$
(depending on $\phi$, $n$, and $T$, but \emph{not} on $N$) such that for every $u\in[0,T]$

\begin{align}
|A_{u}^{N}(\phi)| \le C\big(1+\norma{\delta_{u}^{N}}_{S'}\big).  
\end{align}
Hence, for $0\le s\le t\le T$,

\begin{align}
\nm{\delta^N_t,\varphi}-\nm{\delta^N_s,\varphi}
\le C\int_s^t\big(1+\norma{\delta_{u}^{N}}_{S'}\big)\,\de u.
\end{align}
Square both sides and use Jensen inequality to obtain

\begin{align}
\big|\langle\delta^N_t,\varphi\rangle-\langle\delta^N_s,\varphi\rangle\big|^2
\le C^2 |t-s|\int_s^t\big(1+\|\delta^N_u\|_{S'}\big)^2 \de u
\le C^2 |t-s| \int_0^T \big(1+\|\delta^N_u\|_{S'}\big)^2 \de u.    
\end{align}
Taking the supremum over $s,t$ with $|t-s|<\delta$ and then expectation yields

\begin{align}
\mathbb{E}\!\left[\sup_{|t-s|<\delta}
\big|\langle\delta^N_t,\varphi\rangle-\langle\delta^N_s,\varphi\rangle\big|^2\right]
\le C^2\delta\; \mathbb{E}\!\left[\int_0^T \big(1+\|\delta^N_u\|_{S'}\big)^2 du\right]
\le C' \,\delta \,\Big(1+\mathbb{E}\!\big[\sup_{0\le u\le T}\|\delta^N_u\|_{S'}^2\big]\Big).
\end{align}
By \autoref{auxlem1}and \autoref{auxlem2} the quantity $\sup_{N} \mathbb{E}[\sup_{0\le u\le T}\|\delta_{u}^{N}\|_{S'}^2]$
is finite. Therefore, there is a finite constant $C$ (independent of $N$) such that

\begin{align}
\mathbb{E}\!\left[\sup_{|t-s|<\delta}
\big|\langle\delta^N_t,\varphi\rangle-\langle\delta^N_s,\varphi\rangle\big|^2\right]
\le C\,\delta.    
\end{align}
Applying Markov's inequality to $\sup_{|t-s|<\delta}\big|\langle\delta^N_t,\varphi\rangle-\langle\delta^N_s,\varphi\rangle\big|^2$,
we get for every $N$

\begin{align}
\mathbb{P}\!\Big(\sup_{|t-s|<\delta}
\big|\langle\delta^N_t,\varphi\rangle-\langle\delta^N_s,\varphi\rangle\big|\ge\varepsilon\Big)
\le \frac{C\delta}{\varepsilon^2}.   
\end{align}
Choose $\delta>0$ so small that $C\delta/\varepsilon^2\le\varepsilon$; then the right-hand side is
$\le\varepsilon$ uniformly in $N$. Thus we may take any $M\ge1$ and
conclude

\begin{align}
\sup_{N>M}\; \mathbb{P}\!\Big(\sup_{|t-s|<\delta}
\big|\langle\delta^N_t,\varphi\rangle-\langle\delta^N_s,\varphi\rangle\big|\ge\varepsilon\Big)
\le \varepsilon.   
\end{align}
\end{proof}
In what follows, we prove that \eqref{limiteq} has a unique solution.

\begin{prop}[Existence and Uniqueness for ~\eqref{limiteq}]
Let $\mu_t$ be the solution of the nonlinear continuity equation \eqref{FKP} and let $\delta_{0} \in S'$ be an initial signed measure. 
Consider the linear evolution equation for the fluctuation process \eqref{limitFKP} interpreted in the weak sense: for every $\phi\in S$,
\begin{equation}\label{eq:weak-form}
\frac{\de }{\de t}\nm{\delta_t,\phi}
= \nm{\delta_t,\, \nabla_\theta V(\theta,\mu_t)\!\cdot\!\nabla_\theta\phi}
+ \nm{\mu_t,\, \nabla_\theta G(\theta,\delta_t)\!\cdot\!\nabla_\theta\phi}.
\end{equation}
Then equation~\eqref{limitFKP} admits a unique solution

\begin{align}
\delta_{\cdot}\in C([0,T];(\mathbb{W}(\T^{d}))') \quad\text{such that}\quad \sup_{t\in[0,T]}\norma{\delta_t}_{(\mathbb{W}(\T^{d}))'} < \infty .    
\end{align}
\end{prop}
\begin{proof}
Let $\delta_t^{(1)}$ and $\delta_t^{(2)}$ be two solutions of \eqref{limitFKP} with the same initial condition $\delta_0$,
and set $\xi_{t}\coloneqq\delta_{t}^{(1)} - \delta_{t}^{(2)}$. Then $\xi_{t}$ satisfies

\begin{align}
\partial_{t}\xi_{t}=-\nabla\!\cdot \!\big(\nabla \mathcal{V}(\theta,\mu_{t})\,\xi_{t}+ \mu_{t}\nabla\mathcal{G}(\theta,\xi_{t})\big), \hskip 0,2cm \xi_{0} = 0.    
\end{align}

For every $\phi\in S$ with $\norma{\phi}_{S}=1$ we have

\begin{align}
\frac{\de}{\de t}\nm{\xi_t,\phi}= \nm{\xi_{t}, \nabla \mathcal{V}(\theta,\mu_t)\cdot\nabla\phi}
+ \nm{\mu_{t},\nabla\mathcal{G}(\theta,\xi_{t})\cdot\nabla\phi}.
\end{align}
Using the uniform bounds derived in \eqref{unifb1}-\eqref{unifb2} we obtain 

\begin{align}
\Big|\tfrac{\de }{\de t}\nm{\xi_{t},\phi}\Big|
\le C_1\,\norma{\xi_{t}}_{S'}, \hskip 0,2cm C_{1}\coloneqq nd + 6nd \sqrt{{\rm vol}(\T^{d})}.  
\end{align}

Taking the supremum over all $\phi$ with $\norma{\phi}_{S}=1$ gives

\begin{align}
\frac{\de}{\de t}\norma{\xi_{t}}_{S'}\le C_1\,\norma{\xi_{t}}_{S'}.    
\end{align}
By Grönwall’s inequality and $\xi_0=0$, we deduce
\begin{align}
\norma{\xi_{t}}_{S'}=0    
\end{align}
for all $t\in[0,T]$, hence $\delta_{t}^{(1)}=\delta_{t}^{(2)}$. This proves uniqueness. For existence, write the corresponding integral formulation:

\begin{align}
\nm{\delta_t,\phi}=\nm{\delta_0,\phi}+ \int_{0}^{t}\Big[
\nm{\delta_s, \nabla\mathcal{V}(\theta,\mu_s)\!\cdot\!\nabla_\theta\phi}
+ \nm{\mu_s, \nabla\mathcal{G}(\theta,\delta_s)\!\cdot\!\nabla_\theta\phi}
\Big]\de s.
\end{align}
The map $\delta \mapsto \nm{\delta_s, \nabla\mathcal{V}(\theta,\mu_s)\!\cdot\!\nabla_\theta\phi}
+ \nm{\mu_s, \nabla\mathcal{G}(\theta,\delta_s)\!\cdot\!\nabla_\theta\phi}$ is linear and Lipschitz in $\delta$ on $(\mathbb{W}(\T^{d}))'$
with constant $C_1$. Thus, by the Picard–Lindelöf theorem in the Banach space $C([0,T];(\mathbb{W}(\T^{d}))')$, there exists a unique continuous solution satisfying

\begin{align}
\sup_{t\in[0,T]}\norma{\delta_{t}}_{(\mathbb{W}(\T^{d}))'}
\le e^{C_1T}\norma{\delta_{0}}_{(\mathbb{W}(\T^{d}))'}.    
\end{align}

\end{proof}
\begin{rem}
The uniqueness of~\eqref{limitFKP} implies that any limit point of the sequence $\{\delta_{t}^{N}\}_{N\in\N}$ must satisfy the same deterministic equation with identical initial condition. Consequently, the entire sequence $\delta_{t}^{N}$ converges in law towards the unique solution $\delta_{t}$ of~\eqref{limitFKP}.
\end{rem}
We now are in position to prove \autoref{mainprop1}.
\begin{proof}[{Proof of \autoref{mainprop1}}]
Notice that by \autoref{auxlem3} the trajectories of $\delta_{t}^{N}$ are almost surely continuous in $S'$. Furthermore, by \eqref{decayprob} and \autoref{auxlem5}, we may conclude that the sequence $(\delta_{t}^{N})_{N}$ is relatively compact in the Skorokhod space $\mathcal{D}([0,T],S')$. Notice that by considering \eqref{limiteq}, we may pass to the limit, and then we obtain that $\delta_{t}$ satisfies \eqref{limitFKP} for all $t\in [0,T]$. 
\end{proof}
Let us now state the proof of Corollary \autoref{convergenzadif}.
\begin{proof}[{Proof of \autoref{convergenzadif}}]
Since
\begin{align}
\psi_{N}(t,x)=\int f(\theta,x)\delta_{t}^{N}  
\end{align}
the desired convergence is a direct consequence of \autoref{mainprop1}.
\end{proof}
We now prove item a of \autoref{thm:LDP}. 
\begin{lem}\label{lem:LDP_1}
Let $\mu_{0}\in\cP(\mathbb T^d)$, and $d>4$. Let us denote by $\mu_{\ov\Theta_{\cdot}^{N}}:t\in[0,T]\mapsto \mu_{\ov\Theta_{t}^{N}}$, and $\mu_{\Theta_{\cdot}^{N}}:t\in[0,T]\mapsto \mu_{\Theta_{t}^{N}}$. Then $(\C_{T},\P_{\mu_{\ov\Theta_{\cdot}^{N}}}^{N},N)$, and $(\C_{T},\P_{\mu_{\Theta_{\cdot}^{N}}}^{N},N)$ satisfy a large deviation principle both with same rate functional $L_{\mu_{\cdot}}(\nu)$ where $\nu\in \C_{T}$, and $\mu_{\cdot}$ the unique solution to \eqref{continuity_eq}.
\end{lem}
\begin{proof}
Let us notice that by \cite[Theorem 3.2]{hernandez2025meanfieldlimitgeneralmixtures}, there exists a sequence $(\ov{\Theta}_{t}^{N})_{t}=(\ov{\theta}_{t}^{1,N},\ldots, \ov{\theta}_{t}^{N,N})$ of independent, and identically distributed random variables valued in $\T^{d}$ such that the pathwise propagation of chaos in the sense of \autoref{propchaos} holds true with $p=2$. In particular, we get that $\mu_{\Theta_{\cdot}^{N}}$, and $\mu_{\ov\Theta_{\cdot}^{N}}$ weakly converge to the same limit $\mu_{\cdot}$. Let us now apply \autoref{thm:gen_sanov} with $r=1$, and $X=Y=C([0,T];\T^{d})$. Since $\cP(X)\simeq C([0,T];\cP(\T^{d}))$ our conclusion follows.
\end{proof}
Let us now prove item b of \autoref{thm:LDP}.

\begin{prop}\label{prop:exp-eq}
Let us fix $T>0$, $0<\beta<1$, and $d>4$. Let us suppose that $\Theta_{0}^{N}$ is composed by independent and identically distributed random variables with values in $\mathbb{T}^{d}$. Then for any $\varepsilon>0$, one gets

\begin{align}\label{equiv}
\lim_{N\to\infty}\frac{1}{(\log N)^{\beta}}\log\mathbb P\Big(\de_{T,\cP(\T^{d})}(\mu_{\Theta_{\cdot}^{N}},\mu_{\ov\Theta_{\cdot}^{N}})>\varepsilon\Big)
=-\infty.
\end{align}
\end{prop}

\begin{proof}
Let $\varepsilon>0$. By Markov's inequality one has

\begin{align}
    \P\left[\de_{T,\cP(\T^{d})}(\mu_{\Theta_{\cdot}^{N}},\mu_{\ov\Theta_{\cdot}^{N}})>\varepsilon\right]\leq \frac{1}{\varepsilon^{2}}\E\left[\left(\de_{T,\cP(\T^{d})}(\mu_{\Theta_{\cdot}^{N}},\mu_{\ov\Theta_{\cdot}^{N}})\right)^{2}\right].
\end{align}
On the other hand, by following the proof of the propagation of chaos bound proved in \cite[Theorem 3.2]{hernandez2025meanfieldlimitgeneralmixtures} one gets

\begin{align}
\E\left[\left(\de_{T,\cP(\T^{d})}(\mu_{\Theta_{\cdot}^{N}},\mu_{\ov\Theta_{\cdot}^{N}})\right)^{2}\right]\leq C N^{-\frac{2}{d}}    
\end{align}
for some positive constant $C$ independent of $N$. Then

\begin{align}
\P\left[\de_{T,\cP(\T^{d})}(\mu_{\Theta_{\cdot}^{N}},\mu_{\ov\Theta_{\cdot}^{N}})>\varepsilon\right]\leq \frac{CN^{-\frac{2}{d}}}{\varepsilon^{2}}, 
\end{align}
and from here, we get for large $N$ that
\begin{align}
\frac{1}{(\log(N))^{\beta}}\log\P\left[\de_{T,\cP(\T^{d})}(\mu_{\Theta_{\cdot}^{N}},\mu_{\ov\Theta_{\cdot}^{N}})>\varepsilon\right]\leq (\log(N))^{-\beta}\log(CN^{-\frac{2}{d}})-2\log(\varepsilon)(\log(N))^{-\beta}
\end{align}
By taking the limit on $N$, one gets
\begin{align*}
\lim_{N\rightarrow +\infty}\frac{1}{(\log(N))^{\beta}}\log\P\left[\de_{T,\cP(\T^{d})}(\mu_{\Theta_{\cdot}^{N}},\mu_{\ov\Theta_{\cdot}^{N}})>\varepsilon\right]=-\infty.
\end{align*}
\end{proof}
\section{Conclusions}\label{sec:concl}
In this paper we have studied the fluctuations of mixtures of identical experts given by a quantum neural network. We have obtained a CLT when the parameters follows a gradient flow equation. Within such a framework, the parameters of a single expert are considered as the spatial coordinates of a particle, and the overall training is described by the dynamics of a system of particles induced by a vector field in the sense of \eqref{gradeq3}. As a consequence, a general equation of continuity can be defined. In \autoref{prop:initialcase}, we show that the empirical measure $\mu_{\Theta_{0}^{N}}$ associated to the initial parameters, defined in \eqref{eq:empiricalmeasure}, weakly converges to a Gaussian variable as in the case of the classical CLT for real valued random variables. Furthermore, in \autoref{mainprop1}, we show that the stochastic process $\delta_{t}^{N}$ converges in law to a limiting process $\delta_{t}$ proving that the fluctuations of $\mu_{\Theta_{t}^{N}}$ around its mean field limit $\mu_{t}$ is of order $\frac{1}{\sqrt{N}}$. Differently from \cite{girardi2025,melchor2025quantitative}, the training of the quantum neural network considered here does not happen in the lazy regime, enabling representation learning. Our techniques do not allow us to study the joint limit of infinite depth and width.

Our result open the way to several possible research directions:
\begin{itemize}
    \item A critical next step is determining the convergence of the fluctuation process $\delta_{t}^{N}$ uniformly in time, and to compare this result with a time-uniform upper bounds for the $\W_{2}$ Wasserstein distance between the empirical distribution of the parameters and the limit probability measure. Establishing such time-uniform bounds would be crucial, as it would prove that the mean-field approximation holds even for $t \to \infty$. This is particularly relevant at the end of the training process, when the generated function perfectly reproduces the training examples (i.e., when the loss converges to zero).

    \item Extending our results to the setting where the number of parameters of each expert grows with $N$. This setting would allow us to consider the joint limit of infinite depth and width, and to scale the complexity of each expert with $N$. In this setting, the convergence of the empirical distribution of the parameters is ill-defined, and the probability distribution of the generated function would have to be considered.

    \item A natural avenue for future research is extending our results to the infinite-dimensional limit, where the parameter dimension $d\rightarrow +\infty$. This extension would require a careful analysis of the continuity equation for the limit measure $\mu_{t}$ in the infinite-dimensional space $\R^{\infty}$, following approaches such as those considered in \cite{kolesnikov2014continuity}.
\end{itemize}

\section*{Acknowledgements}
AMH has been supported by project PRIN 2022 ``understanding the LEarning process of QUantum Neural networks (LeQun)'', proposal code 2022WHZ5XH -- CUP J53D23003890006. The author AMH is a member of the ``Gruppo Nazionale per l'Analisi Matematica, la Probabilità e le loro Applicazioni (GNAMPA)'' of the ``Istituto Nazionale di Alta Matematica ``Francesco Severi'' (INdAM)''.

\section*{Declarations}
\noindent
{\bf Data Availability} Authors can confirm that all relevant data are included in the article. \\
\vskip 0,1cm
\noindent
{\bf Conflict of interest} The authors confirm that there is no Conflict of interest.

\bibliographystyle{siam}
\bibliography{bibliography_1}
 \end{document}